\documentclass[11pt]{article}
\usepackage[margin=1in]{geometry}

  \usepackage[utf8]{inputenc}
  \usepackage{amsmath,amsfonts,amsthm,amssymb}
  \usepackage{thmtools}
  \usepackage{xspace}
  \usepackage{xcolor}
  \usepackage{color}
  \usepackage{tikz}
  \usepackage[noline,noend,ruled]{algorithm2e}
  \usepackage{tikz-qtree}
  \usetikzlibrary{decorations.pathmorphing}
  \usetikzlibrary{decorations.markings}
  \usetikzlibrary{shapes.geometric, snakes}
  \usepackage{authblk}
  \usepackage{enumerate}
  
  \usepackage[colorlinks]{hyperref}
  \usepackage[capitalise]{cleveref}

  \theoremstyle{plain}
  \newtheorem{theorem}{Theorem}
  \newtheorem{lemma}[theorem]{Lemma}  
  \newtheorem{corollary}[theorem]{Corollary}

  \newtheorem{observation}[theorem]{Observation}
  \theoremstyle{definition}

  \newtheorem{example}[theorem]{Example}

  \usepackage{mathtools}
\DeclarePairedDelimiter{\floor}{\lfloor}{\rfloor}

\newcommand{\Oh}{\mathcal{O}}
\newcommand{\T}{\mathcal{T}}

\newcommand{\Path}{\mathsf{Path}}
\newcommand{\Triangle}{\mathsf{Triangle}}
\newcommand{\Upper}{\mathsf{Upper}}
\newcommand{\Lower}{\mathsf{Lower}}
\newcommand{\locus}{\mathit{locus}}
\newcommand{\rot}{\mathit{rot}}
\newcommand{\per}{\mathit{per}}
\newcommand{\OOcc}{\mathit{OvOcc}}
\newcommand{\suf}{\mathit{suf}}
\newcommand{\expon}{\mathit{exp}}

\newcommand{\AllPartialCovers}{\textsc{AllPartialCovers}\xspace}
\renewcommand{\c}{\mathit{cv}}
\renewcommand{\o}{\mathit{ov}}
\newcommand{\occ}{\mathit{occ}}
\newcommand{\n}{\mathit{nov}}
\newcommand{\CST}{\mathit{CST}}
\newcommand{\ST}{\mathit{ST}}
\newcommand{\WA}{\mathit{WA}}
\newcommand{\R}{\mathcal{R}}
\newcommand{\Sq}{\mathit{Sq}}

\renewcommand{\output}{\mathsf{output}}

\def\dd{\mathinner{.\,.}}

 \newcommand{\defproblem}[3]{
  \vspace{2mm}
\noindent\fbox{
  \begin{minipage}{0.96\textwidth}
  \textsc{#1}\\
  {\bf{Input:}} #2  \\
  {\bf{Output:}} #3
  \end{minipage}
  }
  \vspace{2mm}
}

\newcommand{\defdsproblem}[4]{
  \vspace{2mm}
\noindent\fbox{
  \begin{minipage}{0.96\textwidth}
  \textsc{#1}\\
  {\bf{Input:}} #2  \\
  {\bf{Query input:}} #3 \\
  {\bf{Query output:}} #4
  \end{minipage}
  }
  \vspace{2mm}
}

\usetikzlibrary{shapes,calc,arrows,positioning,decorations.markings,trees}

\title{Linear Time Construction of Cover Suffix Tree and Applications}

\author[1,2]{Jakub Radoszewski}

\affil[1]{University of Warsaw, Poland, \texttt{jrad@mimuw.edu.pl}}
\affil[2]{Samsung R\&D, Warsaw, Poland}

\date{\vspace{-5ex}}

\begin{document}

\maketitle

\begin{abstract}
The Cover Suffix Tree (CST) of a string $T$ is the suffix tree of $T$ with additional explicit nodes corresponding to halves of square substrings of $T$.
In the CST an explicit node corresponding to a substring $C$ of $T$ is annotated with two numbers:
the number of non-overlapping consecutive occurrences of $C$ and
the total number of positions in $T$ that are covered by occurrences of $C$ in $T$.
Kociumaka et al.\ (\emph{Algorithmica}, 2015) have shown how to compute the CST of a length-$n$ string in $\Oh(n \log n)$ time.
We show how to compute the CST in $\Oh(n)$ time assuming that $T$ is over an integer alphabet.

Kociumaka et al.\ (\emph{Algorithmica}, 2015; \emph{Theor.\ Comput.\ Sci.}, 2018) have shown that knowing the CST of a length-$n$ string $T$,
one can compute a linear-sized representation of all seeds of $T$ as well as all shortest $\alpha$-partial covers and seeds in $T$ for a given $\alpha$ in $\Oh(n)$ time.
Thus our result implies linear-time algorithms computing these notions of quasiperiodicity.
The resulting algorithm computing seeds is substantially different from the previous one (Kociumaka et al., SODA 2012, \emph{ACM Trans.\ Algorithms}, 2020).
Kociumaka et al.\ (\emph{Algorithmica}, 2015) proposed an $\Oh(n \log n)$-time algorithm for computing a shortest $\alpha$-partial cover for each $\alpha=1,\ldots,n$; we improve this complexity to $\Oh(n)$.

Our results are based on a new characterization of consecutive overlapping occurrences of a substring $S$ of $T$ in terms of the set of runs (see Kolpakov and Kucherov, FOCS 1999) in $T$.
This new insight also leads to an $\Oh(n)$-sized index for reporting overlapping consecutive occurrences of a given pattern $P$ of length $m$ in $\Oh(m+\mathsf{output})$ time, where $\mathsf{output}$ is the number of occurrences reported.
In comparison, a general index for reporting bounded-gap consecutive occurrences of Navarro and Thankachan (\emph{Theor.\ Comput.\ Sci.}, 2016) uses $\Oh(n \log n)$ space.

\end{abstract}

\section{Introduction}\label{sec:intro}
The Cover Suffix Tree (CST, in short) of a string $T$, denoted as $\CST(T)$, is the suffix tree of $T$ ($\ST(T)$) augmented with additional nodes and values. For every substring $C$ of $T$, $\CST(T)$ allows to efficiently compute the number of positions in $T$ that are covered by occurrences of $C$, provided that the node representing $C$ in $\CST(T)$ is known. Thus the CST is a generalization of string covers~\cite{DBLP:journals/ipl/ApostolicoFI91}. The CST of a string was introduced by Kociumaka et al.~\cite{DBLP:journals/algorithmica/KociumakaPRRW15} for computing so-called \emph{partial covers} of a string (see below).
Other applications of the CST to the field of quasiperiodicity (see~\cite{DBLP:journals/tcs/ApostolicoE93}) were discussed in \cite{DBLP:journals/algorithmica/KociumakaPRRW15,DBLP:journals/tcs/KociumakaPRRW18a}.

Let $n$ denote the length of a string $T$.
Kociumaka et al.~\cite{DBLP:journals/algorithmica/KociumakaPRRW15} presented an algorithm computing $\CST(T)$ in $\Oh(n \log n)$ time.
Our main result is an algorithm that constructs $\CST(T)$ in $\Oh(n)$ time.
We assume that $T$ is over an \emph{integer alphabet} $\{0,\ldots,n^{\Oh(1)}\}$.
This assumption has become a standard in suffix tree construction algorithms since the linear-time suffix tree construction algorithm of Farach~\cite{DBLP:conf/focs/Farach97}.

In \cref{ss:CST} we provide more details on the CST.
Then in \cref{ss:app1} we discuss applications of our result to computing various notions of quasiperiodicity and in \cref{ss:app2} we present an application of our approach to a variant of text indexing.

\renewcommand{\tabcolsep}{1pt}
\begin{figure}[htpb]
\centering
\begin{tikzpicture}[yscale=1.1]
\draw (0,9) -- node[sloped,right,rotate=90] {
\begin{tabular}{c}
\#
\end{tabular}
} (0.5,8);
\draw (-2,8) -- node[sloped,right,rotate=90] {
\begin{tabular}{c}
\#
\end{tabular}
} (-1.5,6.5);
\draw (-3,7.5) -- node[sloped,left,rotate=270] {
\begin{tabular}{c}
\#
\end{tabular}
} (-4,6);
\draw (-5,7) -- node[sloped,left,rotate=270] {
\begin{tabular}{c}
\#
\end{tabular}
} (-5.5,5.5);

\draw (0,9) -- node[sloped,left,rotate=270] {
\begin{tabular}{c}
a
\end{tabular}
} (-2,8);

\draw (-2,8) -- node[sloped,left,rotate=270] {
\begin{tabular}{c}
a
\end{tabular}
} (-3,7.5);

\draw (-3,7.5) -- node[sloped,left,rotate=270] {
\begin{tabular}{c}
a
\end{tabular}
} (-5,7);

\draw (-5,7) -- node[sloped,left,rotate=270] {
\begin{tabular}{c}
a\\\#
\end{tabular}
} (-6.5,6);

\draw (-5,7) -- node[sloped,right,rotate=90] {
\begin{tabular}{c}
b\\a\\a
\end{tabular}
} (-4.5,5.5);
\draw (-4.5,5.5) -- node[sloped,left,rotate=270] {
\begin{tabular}{c}
a\\a\\\#
\end{tabular}
} (-5,4);
\draw (-4.5,5.5) -- node[sloped,right,rotate=90] {
\begin{tabular}{c}
b\\a\\a\\b\\a\\a\\a\\b\\a\\a\\a\\a\\\#
\end{tabular}
} (-4,-0.5);

\draw (-3,7.5) -- node[sloped,right,rotate=90] {
\begin{tabular}{c}
b
\end{tabular}
} (-2.83333,7);
\draw (-2.83333,7) -- node[sloped,right,rotate=90] {
\begin{tabular}{c}
a
\end{tabular}
} (-2.66666,6.5);
\draw (-2.666666,6.5) -- node[sloped,right,rotate=90] {
\begin{tabular}{c}
a
\end{tabular}
} (-2.5,6);

\draw (-2,8) -- node[sloped,right,rotate=90] {
\begin{tabular}{c}
b\\a
\end{tabular}
} (0,7.33333);
\draw (0,7.33333) -- node[sloped,right,rotate=90] {
\begin{tabular}{c}
a
\end{tabular}
} (1,7);

\draw (0,9) -- node[sloped,right,rotate=90] {
\begin{tabular}{c}
b\\a\\a
\end{tabular}
} (4.5,8);

\foreach \x/\y in {-0.5/0,-4/-1,-7.5/-2}{
\begin{scope}[xshift=\x cm,yshift=\y cm]
\draw (5,8) -- node[sloped,left,rotate=270] {
\begin{tabular}{c}
a
\end{tabular}
} (4.5,7);
\draw (4.5,7) -- node[sloped,left,rotate=270] {
\begin{tabular}{c}
a\\\#
\end{tabular}
} (4,5);
\draw (4.5,7) -- node[sloped,right,rotate=90] {
\begin{tabular}{c}
b\\a\\a\\a\\a\\\#
\end{tabular}
} (5,4);
\draw (5,8) -- node[sloped,right,rotate=90] {
\begin{tabular}{c}
b\\a\\a
\end{tabular}
} (6.5,7);
\draw (6.5,7) -- node[sloped,left,rotate=270] {
\begin{tabular}{c}
a\\b\\a\\a\\a\\a\\\#
\end{tabular}
} (6,4);
\draw (6.5,7) -- node[sloped,right,rotate=90] {
\begin{tabular}{c}
b\\a\\a\\a\\b\\a\\a\\a\\a\\\#
\end{tabular}
} (7,2);
\end{scope}
}

\foreach \x/\y in {0/9,6/7, 0.5/6,2.5/6, -3/5,-1/5,-2.5/6, -5/7,-4.5/5.5}{
    \filldraw (\x,\y) circle (0.04cm);
}
\foreach \x/\y in {4.5/8, 4/7,1/7, -2/8,-3/7.5}{
    \filldraw[white] (\x,\y) circle (0.1cm);
    \draw[blue] (\x,\y) circle (0.1cm);
    \filldraw[white] (\x,\y) circle (0.06cm);
    \filldraw (\x,\y) circle (0.04cm);
}
\foreach \x/\y in {0/7.33333,-2.83333/7,-2.66666/6.5}{
    \filldraw[white] (\x,\y) circle (0.1cm);
    \draw[blue] (\x,\y) circle (0.1cm);
    \filldraw[white] (\x,\y) circle (0.06cm);
}

\draw (6,7) node[above right] {\textbf{9,1}};
\draw (4,7) node[above left] {\textbf{8,2}};
\draw (4.5,8) node[above right] {\textbf{12,4}};
\draw (2.5,6) node[above right=-0.1cm] {\textbf{10,1}};
\draw (0.5,6) node[above left=-0.1cm] {\textbf{9,1}};
\draw (1,7) node[above right] {\textbf{14,2}};
\draw (0,7.33333) node[above] {\textbf{12,4}};
\draw (-1,5) node[above right=-0.1cm] {\textbf{11,1}};
\draw (-3,5) node[above left=-0.1cm] {\textbf{10,1}};
\draw (-2.5,6) node[left] {\textbf{15,1}};
\draw (-2.66666,6.5) node[right] {\textbf{14,2}};
\draw (-2.83333,7) node[right] {\textbf{12,4}};
\draw (-4.5,5.5) node[left] {\textbf{12,2}};
\draw (-5,7) node[above] {\textbf{10,3}};
\draw (-3,7.5) node[above left] {\textbf{14,5}};
\draw (-2,8) node[above left] {\textbf{14,14}};
\end{tikzpicture}
\caption{$\CST(T)$ for $T=\mathrm{aaabaabaabaaabaaaa\#}$. Black circles denote explicit nodes of $\ST(T)$ and blue circles represent nodes corresponding to halves of square substrings in $T$. The numbers next to nodes denote $\c(v)$ and $\n(v)$.}
\label{fig:CST}
\end{figure}
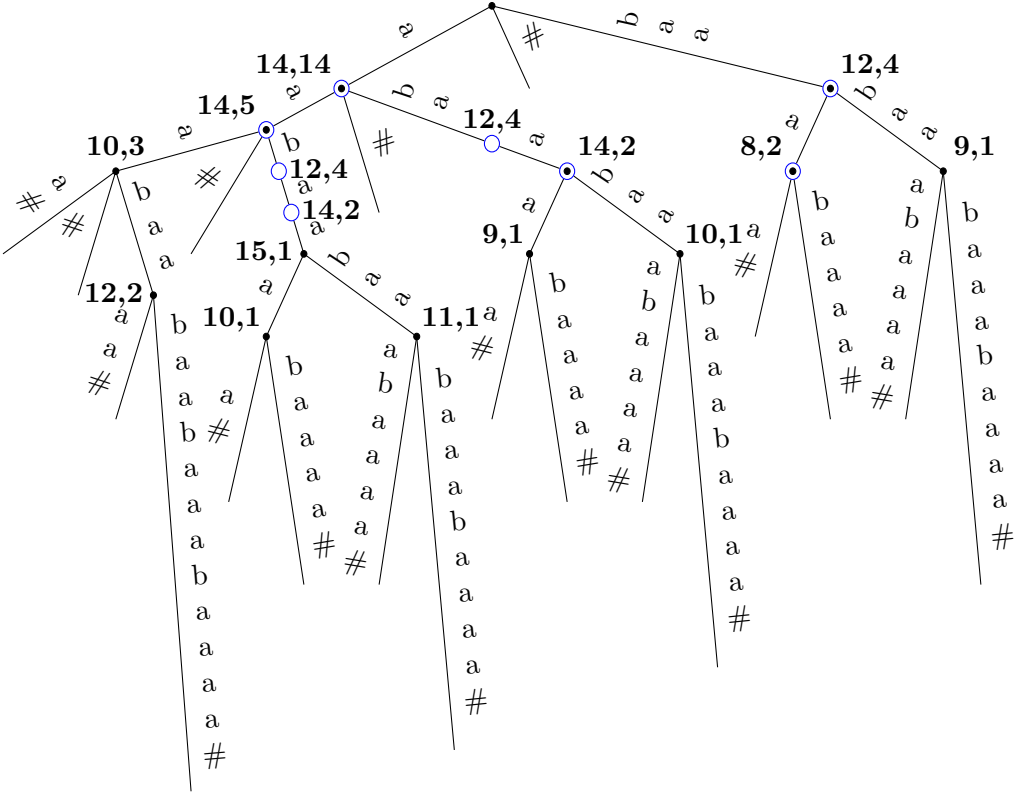
\renewcommand{\tabcolsep}{6pt}

\subsection{Cover Suffix Tree}\label{ss:CST}
In the suffix tree $\ST(T)$, there is a 1-to-1 correspondence between substrings of $T$ and (explicit and implicit) nodes. The same applies to $\CST(T)$. The set of explicit nodes of the $\CST(T)$ comprises of the explicit nodes of the suffix tree of $T$ and of nodes corresponding to halves of square substrings of $T$. Here a \emph{square} is a string of the form $X^2=XX$, for some string $X$ which is called the \emph{square half}.

The CST has the same tree structure as the Maximal Augmented Suffix Tree (MAST) introduced by Apostolico and Preparata for the String Statistics Problem~\cite{DBLP:journals/algorithmica/ApostolicoP96}. It was already observed by Brodal et al.~\cite{DBLP:conf/icalp/BrodalLOP02} that the MAST of $T$ uses only $\Oh(n)$ space; this is because the suffix tree has $\Oh(n)$ nodes \cite{DBLP:journals/jacm/McCreight76} and the number of different square substrings of $T$ is $\Oh(n)$~\cite{DBLP:journals/jct/FraenkelS98}. By a recent result of Brlek and Li~\cite{brlek2022number,DBLP:conf/cwords/BrlekL23} showing that a string of length $n$ contains at most $n$ different square substrings, it follows that the CST (and MAST) contains at most $3n$ explicit nodes.

For a node $v$ of $\CST(T)$, by $\bar{v}$ we denote the substring of $T$ that corresponds to $v$.
A \emph{consecutive occurrence} of string $S$ in $T$ is a pair of indices $(i,j)$ in $T$ such that $j>i$, $S$ has occurrences starting at positions $i$ and $j$ in $T$ and $S$ does not occur at any of the positions $i+1,\ldots,j-1$. A consecutive occurrence $(i,j)$ of $S$ is called \emph{overlapping} if $j < i+|S|$ and \emph{non-overlapping} otherwise.
In $\CST(T)$, each node $v$ is annotated by two values (see \cref{fig:CST}):
\begin{itemize}
    \item $\c(v)$, equal to the total number of positions in $T$ covered by occurrences of $\bar{v}$, and
    \item $\n(v)$, equal to one plus the number of non-overlapping consecutive occurrences of $\bar{v}$ in~$T$. (Intuitively, the one corresponds to the rightmost occurrence of $\bar{v}$ in $T$.)
\end{itemize}

The key property of values $\c(v)$ and $\n(v)$ is that if $u$ is an implicit node of $\CST(T)$ that is located on an edge from an explicit node $v$ to its parent, then $\c(u)$ can be expressed in terms of $\c(v)$ and $\n(v)$ as follows: $\c(u)=\c(v) - (|\bar{v}|-|\bar{u}|)\n(v)$; see \cref{fig:motivation}.

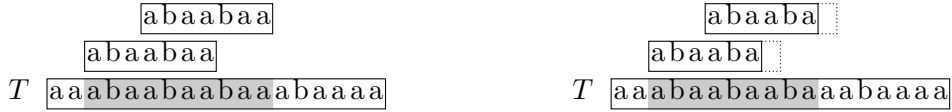
\begin{figure}[htpb]
\centering
\begin{tikzpicture}[xscale=0.25,yscale=0.5]
\draw (-1,0) node[above] {$T$};
\filldraw[white!80!black] (2.5,0) rectangle (12.5,0.8);
\draw (0.5,0) rectangle (18.5,0.8);
\foreach \x/\c in {1/a,2/a,3/a,4/b,5/a,6/a,7/b,8/a,9/a,10/b,11/a,12/a,13/a,14/b,15/a,16/a,17/a,18/a}{
    \draw (\x,0) node[above] {\c};
}
\foreach \dx/\dy in {2/1,5/2}{
    \foreach \x/\c in {1/a,2/b,3/a,4/a,5/b,6/a,7/a}{
        \draw[xshift=\dx cm,yshift=\dy cm] (\x,0) node[above] {\c};
    }
    \draw[xshift=\dx cm,yshift=\dy cm] (0.5,0) rectangle (7.5,0.8);
}
\begin{scope}[xshift=30cm]
\draw (-1,0) node[above] {$T$};
\filldraw[white!80!black] (2.5,0) rectangle (11.5,0.8);
\draw (0.5,0) rectangle (18.5,0.8);
\foreach \x/\c in {1/a,2/a,3/a,4/b,5/a,6/a,7/b,8/a,9/a,10/b,11/a,12/a,13/a,14/b,15/a,16/a,17/a,18/a}{
    \draw (\x,0) node[above] {\c};
}
\foreach \dx/\dy in {2/1,5/2}{
    \foreach \x/\c in {1/a,2/b,3/a,4/a,5/b,6/a}{
        \draw[xshift=\dx cm,yshift=\dy cm] (\x,0) node[above] {\c};
    }
    \draw[xshift=\dx cm,yshift=\dy cm] (0.5,0) rectangle (6.5,0.8);
    \draw[xshift=\dx cm,yshift=\dy cm,densely dotted] (6.5,0) -- (7.5,0) -- (7.5,0.8) -- (6.5,0.8);
}
\end{scope}
\end{tikzpicture}
\caption{Left: the locus $v$ of $C=\mathrm{abaabaa}$ in $\CST(T)$ is an explicit node and we have $\c(v)=10$, $\n(v)=1$ (see \cref{fig:CST}). Right: the locus $v'$ of the prefix $C'=\mathrm{abaaba}$ of $C$ is an implicit node one character above $v$. We then have $\c(v')=\c(v)-\n(v)=9$.}
\label{fig:motivation}
\end{figure}

We obtain the following result.

\begin{restatable}{thm}{thmCST}
\label{thm:CST}
The Cover Suffix Tree (CST) of a string of length $n$ over an integer alphabet can be constructed in $\Oh(n)$ time.
\end{restatable}

The $\Oh(n \log n)$-time algorithm from \cite{DBLP:journals/algorithmica/KociumakaPRRW15} for constructing $\CST(T)$ processes the suffix tree of $T$ bottom-up, storing for each explicit node $v$ the set of occurrences of the corresponding substring of $T$ in an AVL tree. Its time complexity follows by using an efficient algorithm for merging AVL trees~\cite{DBLP:journals/jacm/BrownT79} (cf.~\cite{DBLP:conf/cpm/BrodalP00}). We use a completely different approach, based on a new combinatorial observation that links overlapping consecutive occurrences of substrings of $T$ to runs in $T$~\cite{DBLP:conf/focs/KolpakovK99}. We show that the (multi)set of substrings whose overlapping consecutive occurrences are implied by a run has a simple, ``triangular'' structure. The values $\o(v)$ and $\c(v)$ for explicit nodes $v$ of $\CST(v)$ are then computed in two bottom-up traversals, one in the tree of suffix links of $\ST(T)$ and the other in the $\CST(T)$.

\subsection{Applications of CST to quasiperiodicity}\label{ss:app1}
\cref{thm:CST} has several applications in the field of quasiperiodicity~\cite{DBLP:journals/tcs/ApostolicoE93}. Basic notions of quasiperiodicity are covers~\cite{DBLP:journals/ipl/ApostolicoFI91}, that were already mentioned before, and seeds~\cite{DBLP:journals/algorithmica/IliopoulosMP96}. A string $C$ is a cover of a string $T$ if each position of $T$ is inside at least one occurrence of $C$. A string $S$ is a seed of a string $T$ if it is a cover of a superstring of $T$. In other words, all positions of $T$ are covered by occurrences and overhangs of $S$, where an \emph{overhang} is a prefix of $T$ being a suffix of $S$ or a suffix of $T$ being a prefix of $S$.
A substring $C$ of $T$ is called an \emph{$\alpha$-partial cover} of $T$ if the occurrences of $C$ in $T$ cover at least $\alpha$ positions in $T$. A substring $S$ of $T$ is called an \emph{$\alpha$-partial seed} of $T$ if the occurrences and overhangs of $S$ in $T$ cover at least $\alpha$ positions in $T$.
See \cref{fig:partial} for examples.

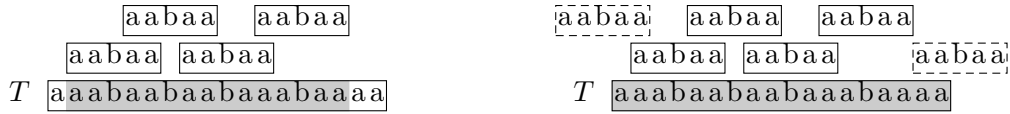
\begin{figure}[htpb]
\centering
\begin{tikzpicture}[xscale=0.25,yscale=0.5]
\draw (-1,0) node[above] {$T$};
\filldraw[white!80!black] (1.5,0) rectangle (16.5,0.8);
\draw (0.5,0) rectangle (18.5,0.8);
\foreach \x/\c in {1/a,2/a,3/a,4/b,5/a,6/a,7/b,8/a,9/a,10/b,11/a,12/a,13/a,14/b,15/a,16/a,17/a,18/a}{
    \draw (\x,0) node[above] {\c};
}
\foreach \dx/\dy in {1/1,4/2,7/1,11/2}{
    \foreach \x/\c in {1/a,2/a,3/b,4/a,5/a}{
        \draw[xshift=\dx cm,yshift=\dy cm] (\x,0) node[above] {\c};
    }
    \draw[xshift=\dx cm,yshift=\dy cm] (0.5,0) rectangle (5.5,0.8);
}
\begin{scope}[xshift=30cm]
\draw (-1,0) node[above] {$T$};
\filldraw[white!80!black] (0.5,0) rectangle (18.5,0.8);
\draw (0.5,0) rectangle (18.5,0.8);
\foreach \x/\c in {1/a,2/a,3/a,4/b,5/a,6/a,7/b,8/a,9/a,10/b,11/a,12/a,13/a,14/b,15/a,16/a,17/a,18/a}{
    \draw (\x,0) node[above] {\c};
}
\foreach \dx/\dy in {1/1,4/2,7/1,11/2}{
    \foreach \x/\c in {1/a,2/a,3/b,4/a,5/a}{
        \draw[xshift=\dx cm,yshift=\dy cm] (\x,0) node[above] {\c};
    }
    \draw[xshift=\dx cm,yshift=\dy cm] (0.5,0) rectangle (5.5,0.8);
}
\foreach \dx/\dy in {-3/2,16/1}{
    \foreach \x/\c in {1/a,2/a,3/b,4/a,5/a}{
        \draw[xshift=\dx cm,yshift=\dy cm] (\x,0) node[above] {\c};
    }
    \draw[xshift=\dx cm,yshift=\dy cm,densely dashed] (0.5,0) rectangle (5.5,0.8);
}
\end{scope}
\end{tikzpicture}
\caption{$C=\mathrm{aabaa}$ is a 15-partial cover of $T$ (left). Indeed, the node $v$ of $\CST(T)$ corresponding to $C$ has $\c(v)=15$ (\cref{fig:CST}). $C$ is also an 18-partial seed of $T$, hence, a seed of $T$ (right).}
\label{fig:partial}
\end{figure}

In \cite{DBLP:journals/algorithmica/KociumakaPRRW15,DBLP:journals/tcs/KociumakaPRRW18a} it was noticed that given the $\CST(T)$, all shortest $\alpha$-partial covers and $\alpha$-partial seeds of $T$ for a specified value of the parameter $\alpha$ can be computed in $\Oh(n)$ time. Thus our \cref{thm:CST} implies the following result.

\begin{corollary}
\label{cor:PCPS}
Let $T$ be a string of length $n$ over an integer alphabet and $\alpha \in [1 \dd n]$. All shortest $\alpha$-partial covers and seeds of $T$ can be computed in $\Oh(n)$ time.
\end{corollary}

In \cite{DBLP:journals/algorithmica/KociumakaPRRW15} also the following problem was considered.

\defproblem{AllPartialCovers}{a string $T$ of length $n$}{for all $\alpha = 1,\ldots,n$, a shortest $\alpha$-partial cover of $T$}

\begin{example}
For $T$ from \cref{fig:CST}, the solution to \AllPartialCovers problem can be as follows: substring $C=\mathrm{a}$ for $\alpha \le 14$, substring $C=\mathrm{aabaa}$ for $\alpha=15$ (see \cref{fig:partial}), and any length-$\alpha$ substring of $T$ for $\alpha \ge 16$.
\end{example}

\noindent
An $\Oh(n \log n)$-time solution for \AllPartialCovers based on $\CST(T)$ and on computing the upper envelope \cite{DBLP:journals/ipl/Hershberger89} of $\Oh(n)$ line segments was presented in \cite{DBLP:journals/algorithmica/KociumakaPRRW15}. We obtain the following result.

\begin{restatable}{thm}{thmPC}
\label{thm:PC}
\AllPartialCovers problem can be solved in $\Oh(n)$ time for a length-$n$ string over an integer alphabet.
\end{restatable}

A linear-time algorithm computing all covers of a string was presented almost 30 years ago by Moore and Smyth~\cite{DBLP:conf/soda/MooreS94,DBLP:journals/ipl/MooreS95}. A rather involved linear-time algorithm computing a representation of all seeds in a string over an integer alphabet was given much more recently by Kociumaka et al.~\cite{DBLP:conf/soda/KociumakaKRRW12}. The representation (already introduced in the earlier, $\Oh(n \log n)$-time algorithm by Iliopoulos et al.~\cite{DBLP:journals/algorithmica/IliopoulosMP96}) consists of a set of paths in the suffix trees of $T$ and of $T$ reversed, at most one path on each edge of the suffix trees. Seeds of $T$ are exactly $|T|$-partial seeds of $T$, so \cref{cor:PCPS} immediately implies an alternative linear-time algorithm computing all shortest seeds of $T$. Moreover, in \cite[Theorem 3]{DBLP:journals/algorithmica/KociumakaPRRW15} it was observed that having $\CST(T)$, the aforementioned representation of all seeds in $T$ can be computed in $\Oh(n)$ time. Thus \cref{thm:CST} yields an alternative $\Oh(n)$-time algorithm computing the same representation of all seeds in $T$ as in~\cite{DBLP:conf/soda/KociumakaKRRW12}. The resulting algorithm is substantially different from the algorithm of~\cite{DBLP:conf/soda/KociumakaKRRW12}; in particular, it is non-recursive and arguably simpler.

Recently, Kociumaka et al.~\cite{DBLP:journals/talg/KociumakaKRRW20} showed that there exists a different representation of all seeds of a string, consisting of $\Oh(n)$ disjoint paths on just the suffix tree of $T$, and that this representation can be computed in $\Oh(n)$ time assuming an integer alphabet. This representation, however, no longer satisfies the convenient property that at most one path on each edge of the suffix tree is in the representation (see \cite[Fig.\ 2]{DBLP:journals/talg/KociumakaKRRW20} for an example).

\subsection{Reporting overlapping occurrences}\label{ss:app2}
Data stored in the CST can be used to compute, for any substring $S$ of $T$, the number $\o(S)$ of overlapping consecutive occurrences of $S$ in $T$. We show that the technique behind \cref{thm:CST} can be further exploited to obtain a linear-space index for reporting overlapping consecutive occurrences of a query pattern.

Navarro and Thankachan \cite{DBLP:journals/tcs/NavarroT16} proposed an index that, given a length-$m$ substring $S$ of $T$ and an interval $[\alpha,\beta]$, reports all consecutive occurrences $(i,j)$ of $S$ such that $\alpha \le j-i \le \beta$ in $\Oh(m+\output)$ time, where $\output$ is the number of consecutive occurrences reported. The size of their index for a text $T$ of length $n$ is $\Oh(n \log n)$. We solve the following problem.

\defdsproblem{Reporting bounded-gap overlapping consecutive occurrences}{A string $T$ of length $n$}{A substring $S$ of $T$, $|S|=m$, and a positive integer $\beta$ such that $\beta < m$}{All consecutive occurrences $(i,j)$ of $S$ in $T$ such that $j-i \le \beta$}

\begin{restatable}{thm}{thmCons}
\label{thm:cons}
There is an index of size $\Oh(n)$ that reports all bounded-gap overlapping consecutive occurrences of a length-$m$ pattern in $\Oh(m+\output)$ time, where $\output$ is the number of consecutive occurrences reported. If $T$ is over a constant-sized alphabet, the index can be constructed in $\Oh(n)$ time. The construction time becomes expected if $T$ is over an integer alphabet.
\end{restatable}

The data structure of \cref{thm:cons} is superior to the data structure of~\cite{DBLP:journals/tcs/NavarroT16} if $\alpha=0$ and $\beta < m$.

In the data structure we use the same combinatorial observation as in \cref{thm:CST}. With it, a query for a pattern $S$ consists in finding the corresponding node $v$ in $\CST(T)$ and reporting all ``triangular'' structures implied by runs that contain $v$. To this end, range minimum query data structures~\cite{DBLP:conf/latin/BenderF00} are used to store the ``bottom sides'' of the triangles. Thanks to this, it is actually sufficient to store $\ST(T)$ instead of the $\CST(T)$.
The expected time in the construction algorithm stems from using perfect hashing~\cite{DBLP:journals/jacm/FredmanKS84} to store children of a node of the suffix tree if $T$ is over a superconstant alphabet.

\subsection{Structure of the paper}
We start by recalling basic definitions related to strings and compact tries (including suffix trees). In \cref{sec:CST} we present the proof of the main \cref{thm:CST}. Solution to \AllPartialCovers (\cref{thm:PC}) is provided in \cref{sec:APC}. The data structure for reporting bounded-gap overlapping consecutive occurrences (proof of \cref{thm:cons}) is presented in \cref{app:cons}. We conclude in \cref{sec:concl}.

\section{Preliminaries}
\subsection{Strings}
By $\Sigma$ we denote the finite alphabet of all the considered strings. We assume that characters of a string $S$ are numbered 1 through $|S|$, with $S[i] \in \Sigma$ denoting the $i$th character. An integer $j \in [1 \dd |S|]$ is called an index in $S$. A string $S[i] S[i+1] \cdots S[j]$ for any indices $i,j$ such that $i \le j$ is called a \emph{substring} of $S$. By $S[i\dd j]$ we denote a \emph{fragment} of $S$ that can be viewed as a positioned substring $S[i] S[i+1] \cdots S[j]$ (formally, it is represented in $\Oh(1)$ space with a reference to $S$ and the interval $[i \dd j]$). We also denote $S[i \dd j-1]$ as $S[i\dd j)$. Two fragments $S[i \dd j]$ and $S[i' \dd j']$ \emph{match} (notation: $S[i \dd j]=S[i' \dd j']$) if the underlying substrings are the same. Similarly we define matching of a fragment and a substring. Two fragments $S[i \dd j]$ and $S[i' \dd j']$ are \emph{equivalent} (notation: $S[i \dd j]\equiv S[i' \dd j']$) if $i=i'$ and $j=j'$. A string $U$ is called a prefix (suffix) of a string $S$ if $U=S[1 \dd |U|]$ ($U=S[|S|-|U|+1 \dd |S|]$, respectively) and a \emph{border} of $S$ if it is both a prefix and a suffix of $S$.

Henceforth by $T$ we denote the text string and by $n$ we denote $|T|$.
We say that a string~$S$ occurs in the text $T$ at position $i$ if $S=T[i \dd i+|S|)$.
A pair of indices $(i,j)$ in~$T$ is called a \emph{consecutive occurrence of substring $S$} if $i<j$, $T[i \dd i+|S|)=T[j \dd j+|S|)$ and $T[k \dd k+|S|) \ne S$ for all $k \in (i \dd j)$. A consecutive occurrence is called \emph{overlapping} if $j<i+|S|$ and otherwise it is called \emph{non-overlapping}. By $\OOcc(S)$ we denote the set of overlapping consecutive occurrences of $S$ in $T$.

For a string $U$ and $d \in \mathbb{Z}_{\ge 0}$, by $U^d$ we denote the $d$th power of $U$, equal to a concatenation of $d$ copies of $U$.
A string $V$ is \emph{primitive} if $V=U^d$ for $d \in \mathbb{Z}_+$ implies that $d=1$.
The following property of primitive strings is a known consequence of Fine and Wilf's lemma~\cite{fine1965uniqueness}.

\begin{lemma}[Synchronization property, see \cite{DBLP:books/daglib/0020103}]\label{lem:synch}
A string $V$ is primitive if and only if $V$ has exactly two occurrences in $V^2$.
\end{lemma}

A string of the form $U^2$ is called a \emph{square}.

\begin{theorem}[\cite{DBLP:journals/jct/FraenkelS98} and \cite{DBLP:journals/tcs/CrochemoreIKRRW14,DBLP:conf/cpm/BannaiIK17}]\label{thm:squares}
The number of distinct square substrings in a length-$n$ string is $\Oh(n)$ and they can all be computed in $\Oh(n)$ time assuming an integer alphabet.
\end{theorem}

We say that a string $S$ has a period $p$ if $S[i]=S[i+p]$ holds for all $i \in [1\dd |S|-p]$; equivalently, if $S$ has a border of length $|S|-p$. By $\per(S)$ we denote the smallest period of~$S$.

A \emph{run} in a string $T$ is a triad $(a,b,p)$ such that (1) $p$ is the smallest period of $T[a\dd b]$, (2) $2p \le b-a+1$, (3) $a=1$ or $T[a-1] \ne T[a-1+p]$, and (4) $b=|T|$ or $T[b+1] \ne T[b+1-p]$. The \emph{exponent} of a run $R=(a,b,p)$ is defined as $\expon(R)=(b-a+1)/p$. By $\R(T)$ we denote the set of all runs in $T$.

\begin{theorem}[\cite{DBLP:journals/siamcomp/BannaiIINTT17}]\label{thm:runs}
A string $T$ of length $n$ has at most $n$ runs and they can be computed in $\Oh(n)$ time if $T$ is over an integer alphabet.
\end{theorem}

An earlier bound $|\R(T)|=\Oh(n)$ together with an $\Oh(n)$-time algorithm for computing $\R(T)$ was proposed in \cite{DBLP:conf/focs/KolpakovK99}. All runs can be computed in $\Oh(n)$ time also for a string over a general ordered alphabet~\cite{DBLP:conf/icalp/Ellert021}.

\subsection{Compact tries}
The \emph{suffix trie} of a string $T$ contains a node for every distinct substring of $T\#$, where $\# \not\in\Sigma$ is a special end marker. The root node is the empty string. For each pair of substrings $(S,Sc)$ of $T$, where $c \in \Sigma$, there is an edge from $S$ to $Sc$ labeled with the character $c$. Each suffix of $T\#$ corresponds to a leaf of the suffix trie.

A \emph{compact suffix trie} of $T$ contains the root, the branching nodes, the leaf nodes, and possibly some other nodes of the suffix trie as \emph{explicit nodes}. Maximal paths that do not contain explicit nodes are replaced by single compact edges, and a fragment of $T$ is used to represent the label of every such edge in $\Oh(1)$ space. The nodes that are dissolved due to compactification are called \emph{implicit nodes}; an implicit node $u$ can be referred to as a pair $(v,d)$ where $v$ is the nearest explicit descendant of $u$ and $d$ is the distance (number of characters) between $u$ and $v$. The most common example of a compact suffix trie of $T$ is the \emph{suffix tree} of $T$, denoted here as $\ST(T)$, in which each maximal branchless path from the suffix trie is replaced by a single compact edge.

\begin{theorem}[\cite{DBLP:conf/focs/Farach97,DBLP:journals/jacm/KarkkainenSB06}]
\label{thm:ST}
The suffix tree of a string of length $n$ over an integer alphabet can be constructed in $\Oh(n)$ time.
\end{theorem}

For a node $v$ of a compact suffix trie $\T$ of $T$, the corresponding substring $\bar{v}$ of $T$ is called the \emph{string label} of $v$.
Conversely, for a substring (or fragment) $S$ of $T$, its locus in $\T$ is the (explicit or implicit) node $v$ of $\T$ such that $\bar{v}=S$.

The locus in $\ST(T)$ of a substring $S$ is denoted as $\locus(S)$.
For a non-root explicit node $v$ of $\ST(T)$, its \emph{suffix link} leads from $v$ to the node $\suf(v)=\locus(X)$, where $\bar{v}=cX$, $c \in \Sigma$; it is known that $\suf(v)$ is then an explicit node.
By $\ST'(T)$ we denote tree of suffix links in $\ST(T)$. The nodes of $\ST'(T)$ are the explicit nodes of $\ST(T)$ and for each non-root explicit node $v$ of $\ST(T)$, in $\ST'(T)$ there is an edge connecting node $v$ with node $\suf(v)$.

Let $\Sq(T)=\{S\,:\,S^2\mbox{ is a substring of }T\}$. Then the \emph{tree structure} of $\CST(T)$ is a compact suffix trie of $T$ that could be obtained from the suffix tree $\ST(T)$ by making loci of substrings $\Sq(T)$ explicit.

A \emph{weighted ancestor query} on a compact suffix trie $\T$ is given a leaf $\ell$ of $\T$ and a non-negative integer $d$ and asks for the topmost (explicit) ancestor $w$ of $\ell$ such that $|\bar{w}| \ge d$. We denote such a query and its result as $w=\WA(\ell,d)$. We use the following offline solution to the problem of answering $\WA$ queries.

\begin{theorem}[\cite{DBLP:journals/talg/KociumakaKRRW20}]\label{thm:WAQ}
Any $q$ weighted ancestor queries on a compact suffix trie with $\Oh(n)$ nodes of a length-$n$ string can be answered in $\Oh(n+q)$ time.
\end{theorem}

Data structures for answering weighted ancestor queries with different complexities are known~\cite{DBLP:journals/talg/AmirLLS07,DBLP:conf/soda/GanardiG22}, also in the special case of the compact suffix trie being the suffix tree~\cite{DBLP:conf/cpm/BelazzouguiKPR21,DBLP:conf/esa/GawrychowskiLN14}.

Let $v$ be a node of a compact suffix trie and $S=\bar{v}$. Then $\c(v)$ is formally defined as
\[\c(v)\,=\,\bigcup\, \{\,[i \dd i+|S|)\,:\, T[i \dd i+|S|)=S\,\}.\]
Moreover, $\n(v)$ equals one plus the number of non-overlapping consecutive occurrences of $S$ in $T$.
$\CST(T)$ stores the values $\c(v)$ and $\n(v)$ for each explicit node $v$. 
By $\occ(v)$ we further denote the total number of occurrences of $\bar{v}$ in $T$. The values $\occ(v)$ for all explicit nodes of a compact suffix trie can be computed bottom-up in linear time, as $\occ(v)$ is the number of leaves in the subtree of node $v$.
By $\o(v)$ we denote the number of overlapping consecutive occurrences of $\bar{v}$ in $T$. We have $\o(v)+\n(v)=\occ(v)$.
We use the notations $\c()$, $\n()$, $\o()$ and $\occ()$ also for substrings of $T$.

\section{Construction of the Cover Suffix Tree}\label{sec:CST}

\subsection{Computing the tree structure}

\begin{lemma}\label{lem:CSTstructure}
The tree structure of the CST of a string $T$ of length $n$ over an integer alphabet can be computed in $\Oh(n)$ time.
\end{lemma}
\begin{proof}
The suffix tree of a string over an integer alphabet can be constructed in $\Oh(n)$ time (\cref{thm:ST}). By \cref{thm:squares}, the set $\Sq(T)$ of square substring halves, each represented as a fragment of $T$, can be computed in $\Oh(n)$ time.

The final step is to make all implicit nodes of the suffix tree that correspond to elements of $\Sq(T)$ explicit. Let $T[i\dd i+2d)$ be a square substring and $\ell$ be the leaf of the suffix tree of $T$ corresponding to the suffix $T[i \dd n]$. Using a weighted ancestor query we can compute a pair $(v,p)$ where $v=\WA(\ell,d)$ is the nearest explicit descendant of the locus $u$ of $T[i\dd i+d)$ and $p$ is the distance between $u$ and $v$. With \cref{thm:WAQ} a batch of $\Oh(n)$ such queries can be answered in $\Oh(n)$ time. Finally, we use Radix Sort to sort the pairs $(v,p)$ (under an arbitrary, fixed order on nodes of the suffix tree) in $\Oh(n)$ time. As a result, the loci of substrings in $\Sq(T)$ are grouped by their nearest explicit descendants, and each group is sorted by decreasing depths. This allows to make all the desired implicit nodes explicit in $\Oh(n)$ time.
\end{proof}

\subsection{Properties of overlapping consecutive occurrences}
If a substring $S$ of $T$ does not have overlapping occurrences in $T$, i.e., $\o(S)=0$, then $\c(S)=\n(S) \cdot |S| = \occ(S) \cdot |S|$ is easy to compute. Hence, below we characterize overlapping consecutive occurrences of substrings. To this end, we use runs.

For indices $1 \le i \le j_1 \le j_2 \le n$, we denote the set of fragments corresponding to a path in the suffix tree of $T$:
\[\Path(i,j_1,j_2) =\{T[i\dd j]\,:\,j \in [j_1 \dd j_2]\}.\]
For a run $R=(a,b,p)$ in $T$, we denote
\[\Triangle(R) = \Triangle(a,b,p) = \bigcup_{i=a}^{b-2p} \Path(i,i+p,b-p).\]
\cref{fig:triangle} gives a graphical motivation for the name of this set of fragments, whereas \cref{fig:triangle2} shows that in some cases the triangle is ``wrapped''.

The following key combinatorial lemma shows that sets $\Triangle(R)$ for $R \in \R(T)$ are sufficient for counting overlapping consecutive occurrences.

\renewcommand{\tabcolsep}{1pt}
\begin{figure}[htpb]
\centering
\begin{tikzpicture}
\draw (3,3.5) -- node[sloped,left,rotate=270] {
\begin{tabular}{c}
a\\b\\c\\d\\e
\end{tabular}
} (0,0);
\draw (3,3.5) -- node[sloped,left,rotate=270] {
\begin{tabular}{c}
b\\c\\d\\e\\a
\end{tabular}
} (2,0);
\draw (3,3.5) -- node[sloped,right,rotate=90] {
\begin{tabular}{c}
c\\d\\e\\a\\b
\end{tabular}
} (4,0);
\draw (3,3.5) -- node[sloped,right,rotate=90] {
\begin{tabular}{c}
d\\e\\a\\b\\c
\end{tabular}
} (6,0);
\draw (0,0)  -- node[left] {a} (0,-0.5);
\draw (0,-0.5) -- node[left] {b} (0,-1);
\draw (0,-1) -- node[left] {c} (0,-1.5);
\draw (0,-1.5) -- node[left] {d} (0,-2);
\begin{scope}[xshift=2cm]
\draw (0,0)  -- node[left] {b} (0,-0.5);
\draw (0,-0.5) -- node[left] {c} (0,-1);
\draw (0,-1) -- node[left] {d} (0,-1.5);
\end{scope}
\begin{scope}[xshift=4cm]
\draw (0,0)  -- node[left] {c} (0,-0.5);
\draw (0,-0.5) -- node[left] {d} (0,-1);
\end{scope}
\begin{scope}[xshift=6cm]
\draw (0,0)  -- node[left] {d} (0,-0.5);
\end{scope}
\draw[-latex,shorten >=0.05cm,shorten <=0.05cm,densely dotted] (0,-2) -- node[sloped,below=0.15cm] {$\suf$} (2,-1.5);
\draw[-latex,shorten >=0.05cm,shorten <=0.05cm,densely dotted] (2,-1.5) -- node[sloped,below=0.15cm] {$\suf$} (4,-1);
\draw[-latex,shorten >=0.05cm,shorten <=0.05cm,densely dotted] (4,-1) -- node[sloped,below=0.15cm] {$\suf$} (6,-0.5);

\foreach \x/\y in {0/-2,2/-1.5,4/-1,6/-0.5, 0/0,2/0,4/0,6/0}{
\filldraw (\x,\y) circle (0.05cm);
}
\foreach \x/\y in {0/-1.5,0/-1,0/-0.5, 2/-1,2/-0.5, 4/-0.5}{
\filldraw (\x,\y) circle (0.03cm);
}

\draw[blue,very thick,rounded corners=1mm] (-0.1,-0.1) rectangle (6.1,0.1);
\draw (6.2,0) node[right] {\textcolor{blue}{$\Upper(R)$}};

\draw[red,very thick,rounded corners=1mm] (-0.0675,-2+0.1125) -- (-0.0675,-2-0.1125) -- (6+0.0675,-0.5-0.1125) -- (6+0.0675,-0.5+0.1125) -- cycle;
\draw (6.2,-0.5) node[right] {\textcolor{red}{$\Lower(R)$}};

\draw[green!70!black,very thick,rounded corners=1mm] (-0.4,-0.4) -- (-0.4,-2.3) -- (6.3,-0.65) -- (6.3,-0.4) -- cycle;
\draw (6.2,-1) node[right] {\textcolor{green!70!black}{$\Triangle(R)$}};

\begin{scope}[xshift=7.5cm,yshift=1.5cm,xscale=0.3]
\draw (-1,0) node[above] {$R$};
\foreach \x/\c in {1/a,2/b,3/c,4/d,5/e, 6/a,7/b,8/c,9/d,10/e, 11/a,12/b,13/c,14/d}{
    \draw (\x,-0.4) node[above] {\tiny \x};
    \draw (\x,0) node[above] {\c};
}
\clip (0.5,0) rectangle (14.5,2);
\foreach \dx in {0.5,5.5,10.5}{
    \draw[xshift=\dx cm] (0,0.5) .. controls (1.5,0.8) and (3.5,0.8) .. (5,0.5);
}
\end{scope}
\end{tikzpicture}
\caption{Illustration of the sets $\Triangle(R)$, $\Upper(R)$ and $\Lower(R)$ on paths in $\CST(T)$ for an example run $R=(1,14,5)$. All nodes representing fragments from $\Triangle(R)$ are distinct because $\expon(R) = 2.8 \le 3$. The nodes in $\Upper(R)$ and $\Lower(R)$ are explicit in $\CST(T)$ (see \cref{obs:lower_upper}).}
\label{fig:triangle}
\end{figure}
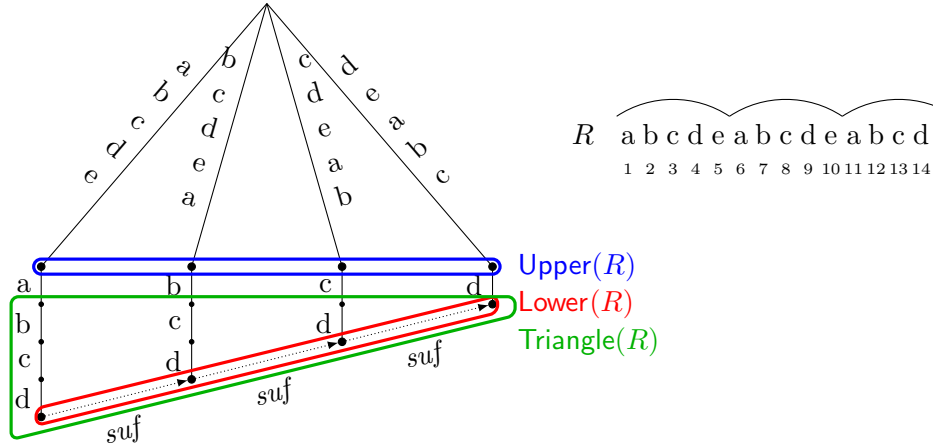
\renewcommand{\tabcolsep}{6pt}

\begin{lemma}\label{lem:triangle}
Let $S$ be a string of length $d$. Then $S$ has an overlapping consecutive occurrence $(i,j)$ in $T$ for some indices $i$, $j$ if and only if $S$ matches a fragment $T[i \dd i+d) \in \Triangle(R)$ for some run $R$ with period $j-i$ in $T$.
\end{lemma}
\begin{proof}
$(\Rightarrow)$ If $S$ has an overlapping consecutive occurrence $(i,j)$, then the substring $F=T[i \dd j+d)$ has a border $S$, so $F$ has a period $p=j-i < d$.

We further have $|F|=j+d-i > 2(j-i)=2p$. Period $p$ is the smallest period of $F$; indeed, a period $q \in [1 \dd p)$ would imply an occurrence $T[i+q \dd i+q+d)$ of $S$ at position $i+q$ such that $i<i+q<i+p=j-i$, so $(i,j)$ would not be a consecutive occurrence of $S$.

Finally, fragment $F$ extends to a unique maximal periodic fragment with smallest period $p$; it is a run $R=(a,b,p)$ in $T$. We have $T[i \dd i+d) \in \Path(i,i+p,b-p) \subseteq \Triangle(R)$ as $i \in [a \dd b-2p]$.

$(\Leftarrow)$ Assume that $T[i \dd i+d) \in \Triangle(R)$ holds for some run $R=(a,b,p)$ and $S=T[i \dd i+d)$. Then $[i \dd i+d+p) \subseteq [a \dd b]$, so the period of the run implies that $T[i+p \dd i+p+d) = T[i\dd i+d)=S$. Moreover, $d>p$ by the definition of $\Triangle(R)$, so the two occurrences of $S$ overlap.

Finally, we need to show that $(i,j)$, for $j=i+p$, is a consecutive occurrence of $S$. If there was an occurrence $T[k \dd k+d)=S$ with $i<k<j$, then the string $X=T[i \dd i+p)$ would have an occurrence in $T[i \dd i+2p)=X^2$ being neither a prefix nor a suffix of $X^2$. String $X$ is primitive, as otherwise the run $R$ would have a period smaller than $p$. Therefore this situation is impossible by the synchronization property (\cref{lem:synch}).
\end{proof}

\renewcommand{\tabcolsep}{1pt}
\begin{figure}[htpb]
\centering
\begin{tikzpicture}
\draw (1,1*1.6) -- node[sloped,left,rotate=270] {
\begin{tabular}{c}
a\\b
\end{tabular}
} (0,0*1.6);
\draw (1,1*1.6) -- node[sloped,right,rotate=90] {
\begin{tabular}{c}
b\\a
\end{tabular}
} (2,0*1.6);
\draw (0,0*1.6)  -- node[left] {a} (0,-0.5*1.6);
\draw (0,-0.5*1.6) -- node[left] {b} (0,-1*1.6);
\draw (0,-1*1.6) -- (0,-1.5*1.6);
\draw (0,-1.3*1.6) node[left] {a};
\draw (0,-1.5*1.6) -- node[left] {b} (0,-2*1.6);
\draw (0,-2*1.6) -- (0,-2.5*1.6);
\draw (0,-2.3*1.6) node[left] {a};
\draw (0,-2.5*1.6) -- node[left] {b} (0,-3*1.6);
\begin{scope}[xshift=2cm]
\draw (0,0*1.6)  -- node[right] {b} (0,-0.5*1.6);
\draw (0,-0.5*1.6) -- node[right] {a} (0,-1*1.6);
\draw (0,-1*1.6) -- node[right] {b} (0,-1.5*1.6);
\draw (0,-1.5*1.6) -- node[right] {a} (0,-2*1.6);
\draw (0,-2*1.6) -- node[right] {b} (0,-2.5*1.6);
\end{scope}
\draw[-latex,shorten >=0.05cm,shorten <=0.05cm,densely dotted] (0,-3*1.6) -- node[sloped,below=0.15cm] {$\suf$} (2,-2.5*1.6);
\draw[-latex,shorten >=0.05cm,shorten <=0.05cm,densely dotted] (2,-2.5*1.6) -- node[sloped,below] {$\suf$} (0,-2*1.6);
\draw[-latex,shorten >=0.05cm,shorten <=0.05cm,densely dotted] (0,-2*1.6) -- node[sloped,below=0.15cm] {$\suf$} (2,-1.5*1.6);
\draw[-latex,shorten >=0.05cm,shorten <=0.05cm,densely dotted] (2,-1.5*1.6) -- node[sloped,below] {$\suf$} (0,-1*1.6);
\draw[-latex,shorten >=0.05cm,shorten <=0.05cm,densely dotted] (0,-1*1.6) -- node[sloped,below=0.15cm] {$\suf$} (2,-0.5*1.6);

\draw[blue,very thick,rounded corners=1mm] (-0.1,-0.1) rectangle (2.1,0.1);
\draw (0,0*1.6) node[above left] {\textcolor{blue}{3x}};
\draw (2,0*1.6) node[above right] {\textcolor{blue}{3x}};
\draw (3,0*1.6) node[right] {\textcolor{blue}{$\Upper(R)$}};

\begin{scope}[xshift=6.5cm,yshift=0cm,xscale=0.3]
\draw (-1,0) node[above] {$R$};
\foreach \x/\c in {1/a,2/b,3/a,4/b,5/a,6/b,7/a,8/b,9/a,10/b}{
    \draw (\x,0) node[above] {\c};
    \draw (\x,-0.4) node[above] {\tiny \x};
}
\foreach \dx in {0.5,2.5,4.5,6.5,8.5}{
    \draw[xshift=\dx cm] (0,0.5) .. controls (0.7,0.7) and (1.3,0.7) .. (2,0.5);
}
\end{scope}

\foreach \x/\y in {0/-3,2/-2.5,0/-2,2/-1.5,0/-1,2/-0.5, 0/0,2/0}{
\filldraw (\x,\y*1.6) circle (0.05cm);
}
\foreach \x/\y in {0/-2.5,0/-1.5,0/-0.5, 2/-2,2/-1}{
\filldraw (\x,\y*1.6) circle (0.03cm);
}

\begin{scope}[yscale=1.6]
\draw[yshift=1cm,red,very thick,rounded corners=1mm] (-0.0675,-2+0.1125) -- (-0.0675,-2-0.1125) -- (2+0.0675,-1.5-0.1125) -- (2+0.0675,-1.5+0.1125) -- cycle;
\draw[red,very thick,rounded corners=1mm] (-0.0675,-2+0.1125) -- (-0.0675,-2-0.1125) -- (2+0.0675,-1.5-0.1125) -- (2+0.0675,-1.5+0.1125) -- cycle;
\draw[yshift=-1cm,red,very thick,rounded corners=1mm] (-0.0675,-2+0.1125) -- (-0.0675,-2-0.1125) -- (2+0.0675,-1.5-0.1125) -- (2+0.0675,-1.5+0.1125) -- cycle;
\draw (3,-0.5) node[right] {\textcolor{red}{$\Lower(R)$}};

\draw[green!70!black,very thick,rounded corners=1mm] (-0.4,-0.4) -- (-0.4,-1.25) -- (2.3,-0.59) -- (2.3,-0.4) -- cycle;
\draw[green!70!black,very thick,rounded corners=1mm] (-0.6,-0.4) -- (-0.6,-2.28) -- (2.5,-1.56) -- (2.5,-0.4) -- cycle;
\draw[green!70!black,very thick,rounded corners=1mm] (-0.8,-0.4) -- (-0.8,-3.32) -- (2.7,-2.53) -- (2.7,-0.4) -- cycle;
\draw (3,-1) node[right] {\textcolor{green!70!black}{$\Triangle(R)$}};
\end{scope}

\end{tikzpicture}
\caption{Illustration of the sets $\Triangle(R)$, $\Upper(R)$ and $\Lower(R)$ on paths in $\CST(T)$ for a run $R=(1,10,2)$ with exponent 5. The substrings in the set $\Triangle(R)$ form a multiset being the sum of the two trapezia and a triangle. The set $\Upper(R)$ contains six fragments; three of them match substring ab, and the remaining three match ba.}
\label{fig:triangle2}
\end{figure}
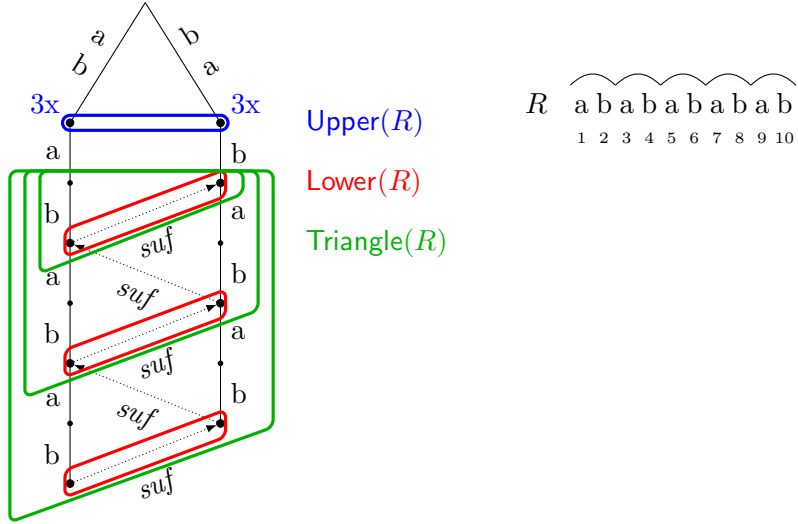
\renewcommand{\tabcolsep}{6pt}

For a run $R=(a,b,p)$, we further denote:
\begin{align*}
    \Upper(a,b,p)&=\{T[i \dd i+p)\,:\,i \in [a \dd b-2p]\}\\
    \Lower(a,b,p)&=\{T[i \dd b-p]\,:\,i \in [a \dd b-2p]\}
\end{align*}
Intuitively, $\Lower(R)$ consists of bottommost endpoints of paths $\Path$ from $\Triangle(R)$, whereas $\Upper(R)$ consists of parents of topmost endpoints of these paths. Informally, they are the ``lower side'' and the ``excluded upper side'' of the triangle; see also \cref{fig:triangle,fig:triangle2}. Below we show basic properties of these sets.

\begin{observation}\label{obs:lower_upper}
Let $R$ be a run in $T$.
\begin{enumerate}[(a)]
\item\label{ita} All fragments in $\Upper(R)$ are square halves in $T$.
\item\label{itb} The loci of fragments in $\Lower(R)$ are explicit nodes in $\ST(T)$.
\end{enumerate}
\end{observation}
\begin{proof}
\eqref{ita} By the periodicity of run $R$, each fragment $T[i \dd i+p) \in \Upper(R)$ is followed by a matching fragment $T[i+p \dd i+2p)$. This is because $i \le b-2p$. Hence, $T[i \dd i+2p)$ is indeed a square in $T$.

\eqref{itb} Let $T[i \dd b-p] \in \Lower(R)$ and $c=T[b+1-p]$. The period of the run implies that $T[i \dd b-p]=T[i+p \dd b]$.

If $T[i \dd b-p]$ is a suffix of $T$, its locus in $\ST(T)$ is explicit as the locus has children along the characters $c$ and $\#$.

Otherwise, character $c'=T[b+1]$ is different from $c$ by the right maximality of the run $R$. Hence, $T[i \dd b-p]c$ and $T[i \dd b-p]c'$ are different substrings of $T$, as claimed.
\end{proof}

Let us note that if $\expon(R)>3$, then each of the sets $\Upper(R)$, $\Triangle(R)$ may contain matching fragments; see \cref{fig:triangle2}.

\subsection{Counting overlapping consecutive occurrences}\label{ss:o}
For each explicit node $v$ of $\CST(T)$, instead of $\n(v)$, we will compute the number $\o(v)$ of overlapping consecutive occurrences of the substring $\bar{v}$ in $T$.

For a set $\mathcal{F}$ of fragments of $T$, we denote
by $\#_v(\mathcal{F})=|\{T[i \dd j] \in \mathcal{F}\,:\,\bar{v }=T[i \dd j]\}|$ the number of fragments in $\mathcal{F}$ that match $\bar{v}$.
\cref{lem:triangle} implies the following formula for $\o(v)$.
\begin{observation}\label{obs:ovTriangle}
For a node $v$ of $\CST(T)$,
$\o(v)=\sum_{R \in \R(T)} \#_v(\Triangle(R))$.
\end{observation}
We will show how to efficiently evaluate these formulas for all explicit nodes $v$ simultaneously.

For each explicit node $v$ of $\CST(T)$ we will compute two counters:
\[C_{upper}[v]=\sum_{R \in \R(T)} \#_v(\Upper(R)),\quad
C_{lower}[v]=\sum_{R \in \R(T)} \#_v(\Lower(R)).\]
That is, $C_{upper}[v]$ ($C_{lower}[v]$) stores the number of times fragments matching the substring $\bar{v}$ occur in $\Upper(R)$ ($\Lower(R)$, respectively) over all runs $R \in \R(T)$.

\newcommand{\subtree}{\mathit{subtree}}
\newcommand{\rroot}{\mathit{root}}
For an explicit node $v$ of $\CST(T)$, by $\subtree(v)$ we denote the set of explicit descendants of $v$ in the tree (including $v$).
The following lemma shows how to compute $\o$ from the counters $C_{upper}$ and $C_{lower}$.
The lemma follows by \cref{obs:ovTriangle}.

\allowdisplaybreaks
\begin{lemma}\label{lem:formula}
For an explicit node $v$ of the $\CST(T)$, we have
\[\o(v)=\sum_{w \in \subtree(v)}(C_{lower}[w]-C_{upper}[w]).\]
\end{lemma}
\begin{proof}
If node $x$ is an ancestor of node $y$, by $x \rightsquigarrow y$ we denote the set of explicit nodes on the path from $x$ to $y$.
By $\rroot$ we denote the root of $\CST(T)$.
By the definitions of $\Triangle(R)$, $\Upper(R)$ and $\Lower(R)$ and \cref{obs:ovTriangle}, we have:
\begin{align*}
\o(v)&=\sum_{R \in \R(T)} \#_v(\Triangle(R))\\
&= \sum_{(a,p,b) \in \R(T)} \sum_{i=a}^{b-2p} \#_v(\Path(i,i+p,b-p))\\
&= \sum_{(a,p,b) \in \R(T)} |\{i \in [a \dd b-2p]\,:\, v \in  (\locus(T[i \dd i+p]) \rightsquigarrow \locus(T[i \dd b-p]))\}|\\
&= \sum_{(a,p,b) \in \R(T)} |\{i \in [a \dd b-2p]\,:\, v \in  (\rroot \rightsquigarrow \locus(T[i \dd b-p]))\}|\\
&\quad -\sum_{(a,p,b) \in \R(T)} |\{i \in [a \dd b-2p]\,:\, v \in  (\rroot \rightsquigarrow \locus(T[i \dd i+p)) \}|\\
&= \sum_{R \in \R(T)} \sum_{w \in \subtree(v)} \#_w(\Lower(R)) - \sum_{R \in \R(T)} \sum_{w \in \subtree(v)} \#_w(\Upper(R))\\
&=\sum_{w \in \subtree(v)} (C_{lower}[w]-C_{upper}[w]).\qedhere
\end{align*}
\end{proof}

\subsubsection{Computing $C_{lower}$ and $C_{upper}$}
Let us recall that $\ST'(T)$ is the tree of suffix links of $\ST(T)$.

\begin{observation}\label{obs:lower}
For each run $R$ in $T$, $\Lower(R)$ forms a path in $\ST'(T)$.
\end{observation}

\begin{lemma}\label{lem:clower}
The counters $C_{lower}[v]$ for all explicit nodes $v$ of $\CST(T)$ can be computed in $\Oh(n)$ time.
\end{lemma}
\begin{proof}
By the observation, $C_{lower}[v]$ is simply the number of paths $\Lower(R)$ that cover node $v$ in $\ST'(T)$ (in particular, no two fragments in a single set $\Lower(R)$ match).

To count paths covering each node in a rooted tree we apply a standard approach using $\pm1$ counters. Initially all counters $C_{lower}[v]$ are equal to 0. For each run $R=(a,b,p) \in \R(T)$, we increment $C_{lower}[v]$ for the bottom endpoint $v$ of the path $\Lower(R)$ ($v=\locus(T[a \dd b-p])$) and decrement $C_{lower}[u]$ for the parent $u$ in $\ST'(T)$ of the top endpoint of $\Lower(R)$ ($u=\locus(T[b-2p+1 \dd b-p])$). In the end for each node $u$ of $\ST'(T)$ in a bottom-up order, we add $C_{lower}[v]$ to $C_{lower}[u]$ for all children $v$ of $u$ in $\ST'(T)$.

Let us summarize and analyze the complexity of the algorithm. Tree $\ST'(T)$ has $\Oh(n)$ nodes. By \cref{thm:runs}, there are at most $n$ paths $\Lower(R)$ and all runs $R$ can be computed in $\Oh(n)$ time. The endpoints of all paths $\Lower(R)$ can be located in $\ST'(T)$ in $\Oh(n)$ time using weighted ancestor queries in $\ST(T)$ (\cref{thm:WAQ}). Finally, the bottom-up traversal of the tree $\ST'(T)$ takes $\Oh(n)$ time. 
\end{proof}

We proceed to computing counters $C_{upper}$.
Let us define an operation $\rot$ such that $\rot(cX)=Xc$ for a string $X$ and character $c \in \Sigma$.
For $k \in \mathbb{Z}_{\ge 0}$, by $\rot^k(S)$ we denote the composition of $\rot$ $k$ times. If $S'=\rot^k(S)$ for some strings $S,S'$ and $k \in \mathbb{Z}_{\ge 0}$, we say that $S'$ is a \emph{cyclic rotation} of $S$. We also say that $S$ and $S'$ are \emph{cyclically equivalent}.

For each run $R$ in $T$, the strings in $\Upper(R)$ are cyclic rotations of each other.
This motivates introduction of the following directed graph $G=(V,E)$.
The set of vertices is $V=\Sq(T)$ and the arcs are defined as follows: $(S,S') \in E$ if and only if $S,S' \in V$ and $S'=\rot(S)$. Instead of addressing vertices of $G$ by substrings of $T$, we will address them by their loci in $\CST(T)$ which are explicit nodes of $\CST(V)$. 

\begin{observation}\label{obs:upper_walk}
For each run $R$ in $T$, $\Upper(R)$ corresponds to a (directed) walk in $G$.
\end{observation}

Let us note that the vertices (and arcs) on the walk $\Upper(R)$ may repeat if $\expon(R)>3$ (see \cref{fig:triangle2} again). In particular, in this case $\Upper(R)$ is contained in a cycle in $G$.

We proceed to the construction of graph $G$. More precisely, a sufficient subset of arcs of $G$ is constructed.

\begin{lemma}\label{lem:G}
A subset $E'$ of $E$ containing all arcs that belong to any walk $\Upper(R)$, for $R \in \R(T)$, can be constructed in $\Oh(n)$ time.
\end{lemma}
\begin{proof}
For each distinct square substring $T[i \dd i+2d)$ of $T$, we insert into $E'$ an arc that leads from $\locus(T[i \dd i+d))$ to $\locus(T[i+1 \dd i+d])$ if $T[i+1 \dd i+d] \in \Sq(T)$; the latter condition can be checked from the construction of the tree structure of $\CST(T)$ (\cref{lem:CSTstructure}). By \cref{thm:squares}, square substrings of $T$ can be enumerated in $\Oh(n)$ time. Then we use off-line weighted ancestor queries (\cref{thm:WAQ}) on $\CST(T)$ to find the desired loci. This concludes that the time complexity is $\Oh(n)$. Now let us argue for the correctness of this algorithm in two steps.

Why $E' \subseteq E$: When adding an arc from $\locus(T[i \dd i+d))$ to $\locus(T[i+1 \dd i+d])$, we know that $T[i \dd i+d) \in \Sq(T)$ and we check if $T[i+1 \dd i+d] \in \Sq(T)$. Hence, an arc connects two vertices of $V=\Sq(T)$. Finally, we have
$\rot(T[i \dd i+d))=T[i+1 \dd i+d]$
because $T[i \dd i+2d)$ is a square. Consequently, $E' \subseteq E$.

Why all arcs that belong to any walk $\Upper(R)$ are in $E'$: Let $T[i \dd i+p),T[i+1 \dd i+p] \in \Upper(a,b,p)$ be two consecutive elements. Then $a \le i < b-2p$, so $T[i \dd i+2p)$ is a square. Therefore, $(\locus(T[i \dd i+p)),\locus(T[i+1 \dd i+p])) \in E'$ by definition.
\end{proof}

In the lemma above, it can be the case that $E' \subsetneq E$ if there are two substrings $S,S' \in \Sq(T)$ such that $S'=\rot(S)$ but there are no two fragments $T[i \dd i+|S|)$, $T[i+1 \dd i+1+|S|)$ matching $S$ and $S'$, respectively. Moreover, it can happen that $E'$ contains an arc that does not belong to any walk $\Upper(R)$. Indeed, when an arc from $\locus(T[i \dd i+d))$ to $\locus(T[i+1 \dd i+d])$ is added to $E'$, we avoid the unnecessary check if $T[i \dd i+d)$, $T[i+1 \dd i+d]$ belong to a set $\Upper(R)$ for any run $R$.

\newcommand{\nnext}{\mathit{next}}
\newcommand{\id}{\mathit{id}}
\begin{lemma}\label{lem:cupper}
The counters $C_{upper}[v]$ for all explicit nodes $v$ of $\CST(T)$ can be computed in $\Oh(n)$ time.
\end{lemma}
\begin{proof}
By \cref{obs:upper_walk}, $C_{upper}[v]$ is the total number of times that walks $\Upper(R)$ visit the node $v \in V$. We will be able to compute these counters efficiently using the fact that graph $G$ has a particularly simple structure: it is a collection of disjoint cycles and paths. The same applies to the graph $G'=(V,E')$ that is computed in \cref{lem:G}. For a node $u$, by $\nnext(u)$ we denote the unique node $v$ such that $(u,v) \in E'$, and $\bot$ if no such node exists.

We can find all cycles in $G'$ using the DFS. Then we apply an algorithm using $\pm 1$ counters (as in the proof of \cref{lem:clower}) and additional counters $C'$ assigned to cycles. Initially all counters are equal to 0. For each cycle $Q$, let us order the nodes $v_1,\ldots,v_{|Q|} \in Q$ along the cycle (arbitrarily) and assign them consecutive id numbers $\id(v_i)=i$. 

For each walk $\Upper(R)$, $R=(a,b,p) \in \R(T)$, we check if its endpoints $v_1$ and $v_2$ ($v_1=\locus(T[a \dd a+p))$ and $v_2=\locus(T[b-2p \dd b-p))$) belong to a cycle. If not, we increment $C_{upper}[v_1]$; we also decrement $C_{upper}[\nnext(v_2)]$ if $\nnext(v_2) \ne \bot$. Otherwise, if $v_1,v_2$ belong to a cycle $Q$, we increase the cycle counter $C'[Q]$ by $\floor{\,|\Upper(R)|\, /\, |Q|\,}$. Moreover (for $v_1,v_2 \in Q$), if $\id(v_1) \le \id(v_2)$, we increment $C_{upper}[v_1]$ and decrement $C_{upper}[\nnext(v_2)]$ if $\id(v_2) < |Q|$. If, however, $\id(v_1) > \id(v_2)$, we increment $C_{upper}[v_1]$ and $C_{upper}(u)$ for the node $u \in Q$ with $\id(u)=1$ and decrement $C_{upper}[\nnext(v_2)]$, thus ``breaking the cyclicity''.

For each cycle $Q$, let us remove the arc $(v,\nnext(v))$ for vertex $v \in Q$ such that $\id(v)=|Q|$. This way $G'$ becomes acyclic; it can be viewed as a forest in which each tree is a path. For each node $v$ of the modified graph $G'$ in topological order, we add $C_{upper}[v]$ to $C_{upper}[\nnext(v)]$ if $\nnext(v) \ne \bot$. Finally, for each original cycle $Q$ in $G'$, we increase $C_{upper}[v]$ for all vertices $v \in Q$ by the counter $C'[Q]$. This way we have computed all the counters $C_{upper}$ as desired.

Let us analyze the complexity.
By \cref{lem:G}, graph $G'=(V,E')$ can be constructed in $\Oh(n)$ time. By \cref{thm:runs}, there are at most $n$ paths $\Upper(R)$ and all runs $R$ can be computed in $\Oh(n)$ time. The endpoints of walks $\Upper(R)$ can be located in $G'$ in $\Oh(n)$ time using weighted ancestor queries in $\CST(T)$ (\cref{thm:WAQ}). Finally, the computation of counters via DFS and topological ordering takes $\Oh(n)$ time. 
\end{proof}

\noindent
This concludes efficient computation of the numbers of overlapping and non-overlapping consecutive occurrences.

\begin{lemma}\label{lem:ov}
Values $\o(v)$ and $\n(v)$ for all explicit nodes $v$ of $\CST(T)$ can be computed in $\Oh(n)$ time.
\end{lemma}
\begin{proof}
We compute the counters $C_{lower}$ and $C_{upper}$ using \cref{lem:clower,lem:cupper}, respectively. By the formula from \cref{lem:formula}, for each node $v$ of $\CST(T)$ in the bottom-up order, $\o(v)$ can be computed as a sum of
$C_{lower}[v]-C_{upper}[v]$ and the sum of values $\o(w)$ for all explicit children $w$ of $v$. Such values can be computed via a bottom-up traversal in $\Oh(n)$ time.

Finally we recall that $\n(v)=\occ(v)-\o(v)$ and that $\occ(v)$ for all explicit nodes of $\CST(T)$ can be easily computed in $\Oh(n)$ time bottom-up.
\end{proof}

\newcommand{\cvov}{\mathit{cv\_ov}}
\newcommand{\cvnov}{\mathit{cv\_nov}}
\subsection{Computing coverage}
For a substring $S$ of $T$, we introduce the following notations:
\[\cvov(S) = \sum_{(i,j) \in \OOcc(S)} (j-i), \quad \cvnov(S) = \n(S) \cdot |S|.\]
As before, we denote $\cvov(v)=\cvov(\bar{v})$ and $\cvnov(v)=\cvnov(\bar{v})$ for nodes $v$ of $\CST(T)$.
The proof of the following observation provides intuition on these definitions.

\begin{observation}
For every substring $S$ of $T$, $\c(S)=\cvov(S)+\cvnov(S)$.
\end{observation}
\begin{proof}
Let us assign each position $k$ of $T$ that is covered by an occurrence of $S$ to the rightmost occurrence $T[i \dd i+|S|)$ of $S$ with $i<k$. Let $j$ be the next occurrence of $S$ to the right of position $i$ (then $j>k$), if any. If $j$ exists and $(i,j) \in \OOcc(S)$, then position $k$ is counted in $\cvov(S)$. Otherwise position $k$ is counted in $\cvnov(S)$.
\end{proof}

Values $\cvnov(v)$ for explicit nodes $v$ of $\CST(T)$ can be easily computed using values $\n(v)$ computed in \cref{lem:ov}. \cref{lem:triangle} yields the following formula for $\cvov(v)$.
\begin{observation}
For a node $v$ of $\CST(T)$,
$\cvov(v)=\sum_{R \in \R(T)} \#_v(\Triangle(R)) \cdot \per(R)$.
\end{observation}

\noindent
Now $\cvov$ values can be computed similarly as $\o$ values were computed in \cref{ss:o}. We just need to multiply counter updates by periods of respective runs.

\begin{lemma}\label{lem:cvov}
The values $\cvov(v)$ for all explicit nodes $v$ of $\CST(T)$ can be computed in $\Oh(n)$ time.
\end{lemma}
\begin{proof}
Let
\[C'_{upper}[v]=\sum_{R \in \R(T)} \#_v(\Upper(R))\cdot \per(R),\quad
C'_{lower}[v]=\sum_{R \in \R(T)} \#_v(\Lower(R))\cdot \per(R).\]
Following the proof of \cref{lem:formula} it can be readily verified that for every explicit node $v$,
\[\cvov(v)=\sum_{w \in \subtree(v)}(C'_{lower}[w]-C'_{upper}[w]).\]
The counters $C'_{lower}[v]$ ($C'_{upper}[v]$) for all explicit nodes can be computed as in \cref{lem:clower} (\cref{lem:cupper}, respectively), where instead of $\pm 1$ counters, for each path $\Lower(R)$ (walk $\Upper(R)$, respectively), $R \in \R(T)$, we add and subtract $\per(R)$ in the respective nodes (and increase cycle counters $C'[Q]$ by amounts $\floor{\,|\Upper(R)|\, /\, |Q|\,} \cdot \per(R)$).
\end{proof}

\noindent
This concludes the construction of $\CST(T)$.
\thmCST*
\begin{proof}
\cref{lem:CSTstructure} can be used to construct the tree structure of $\CST(T)$. We compute $\occ(v)$ for all explicit nodes of $\CST(T)$ in a bottom-up traversal and $\o(v)$ using \cref{lem:ov}, which lets us compute $\n(v)=\occ(v)-\o(v)$ for all explicit nodes. Then for all explicit nodes we compute $\cvov(v)$ using \cref{lem:cvov}, which lets us compute $\c(v)$ for all explicit nodes using values $\cvnov(v)$ that, in turn, depend on $\n(v)$. Each of the lemmas, as well as the bottom-up processing, requires $\Oh(n)$ time.
\end{proof}

\section{Solution to \AllPartialCovers}\label{sec:APC}

An $\Oh(n \log n)$-time solution to \AllPartialCovers from \cite{DBLP:journals/algorithmica/KociumakaPRRW15} is based on computing the upper envelope of $\Oh(n)$ line segments, each connecting points $(|v|,\c(v))$ and $(|v|-k,\c(v)-k\cdot\n(v))$ constructed for an edge of $\CST(T)$ from $v$ to its parent $v'$ containing $k$ implicit nodes. An upper envelope of $\Oh(n)$ line segments can be computed in $\Oh(n \log n)$ time \cite{DBLP:journals/ipl/Hershberger89}.

We show that the \AllPartialCovers problem can be solved in $\Oh(n)$ time using the following observation.
A substring $C$ of $T$ is called \emph{branching} if the locus of $C$ in $\ST(T)$ is a branching node.

\begin{lemma}\label{lem:explicit}
If $C$ is a substring of $T$, then there is a substring $C'$ of $T$ such that $|C'|=|C|$, $\c(C')\ge \c(C)$ and $C'$ is branching or a suffix of $T$.
\end{lemma}
\begin{proof}
Let $C_0=C$.
If $C_0$ is branching or $C_0$ is a suffix of $T$, we are done.
Otherwise, all occurrences of $C_0$ in $T$ are followed by the same character.
Let $a$ be this character, $X$ be $C$ without its first character, and $C_1=Xa$.
We have $|C_1|=|C_0|$ and $\c(C_1)\ge \c(C_0)$.
We use this construction to obtain substrings $C_1,C_2,\ldots$
The sequence ends at the first substring that is branching or a suffix of $T$; such a substring exists since the rightmost occurrence of $C_i$ in $T$, for $i \ge 1$,
is located to the right of the rightmost occurrence of $C_{i-1}$ in $T$.
Let $C_k$, for $k \ge 1$, be the last substring in this sequence.
By the construction, $|C_k|=|C_0|=|C|$, $\c(C_k)\ge \c(C)$ and $C_k$ is branching or a suffix of $T$.
We choose $C'=C_k$.
\end{proof}

\newcommand{\shortest}{\mathit{shortest}}
By the lemma, in the solution to \AllPartialCovers it suffices to iterate over all suffixes of $T$ and branching nodes of $\ST(T)$.
We obtain the following result that was already stated in \cref{sec:intro}.

\begin{algorithm}
\lFor{$i:=1$ \KwSty{to} $n$}{$\shortest[n-i+1] := T[i \dd n]$}
\ForEach{branching node $v$ of $\ST(T)$}{
    \If{$|\bar{v}|<|\shortest[\c(v)]|$}{
        $\shortest[\c(v)]:=\bar{v}$\;
    }
}
\For{$i:=n-1$ \KwSty{down to} $1$}{
    \If{$|\shortest[i]| > |\shortest[i+1]|$}{
        $\shortest[i]:=\shortest[i+1]$\;
    }
}
\caption{Solution to \AllPartialCovers.}\label{algo:APC}
\end{algorithm}

\thmPC*

\begin{proof}
We apply \cref{algo:APC}.
Clearly, the algorithm works in $\Oh(n)$ time.
Let us argue for its correctness.

In the algorithm an auxiliary array $\shortest$ is used that stores fragments of $T$ represented in $\Oh(1)$ space each.
By \cref{lem:explicit}, after the foreach-loop, $\shortest[\alpha]=C$ for $\alpha \in [1 \dd n]$ if $C$ is a shortest substring of $T$ such that $\c(C) = \alpha$.
At the end, $\shortest[\alpha]=C$ if $C$ is a shortest substring of $T$ such that $\c(C) \ge \alpha$.
Hence, $\shortest[\alpha]$ is a shortest $\alpha$-partial cover of $T$ by definition.
\end{proof}

\newcommand{\MinLower}{\mathit{MinLower}}
\newcommand{\Bottom}{\textsf{Bottom}}
\newcommand{\MinBottom}{\mathit{MinBottom}}
\newcommand{\Bottoms}{\mathit{Bottoms}}
\newcommand{\ML}{\mathit{ML}}
\newcommand{\MB}{\mathit{MB}}
\section{Reporting bounded-gap overlapping consecutive occurrences}\label{app:cons}
This section is devoted to the proof of \cref{thm:cons}. We first describe the data structure and then the query algorithm. Finally, we show how the data structure can be constructed efficiently.

The idea behind the data structure is as follows. By \cref{lem:triangle}, an overlapping consecutive occurrence $(i,j)$ of $S$ in $T$ corresponds to a fragment $T[i \dd i+|S|)$ that belongs to a set $\Triangle(R)$ for some run $R$ in $T$; here $j=i+\per(R)$.

Assume first that $\beta=|S|-1$ holds in a query. Given $S$, we would like to list all runs $R$ and positions $i$ such that $T[i \dd i+|S|) \in \Triangle(R)$; then we will report consecutive occurrences of the form $(i,i+\per(R))$. We can mark in $\ST(T)$ all explicit nodes $v$ such that $\bar{v}$ matches a fragment from $\bigcup_R \Lower(R)$, called the ``lower nodes'', and all explicit nodes $v$ such that $\bar{v}$ matches a longest fragment in some $\Triangle(R)$, called the ``bottom nodes''. Now given $S$, we find all explicit descendants $v$ of the locus of $S$ that are ``lower nodes''. Finally, for each $v$, we find all its descendants $w$ in $\ST'(T)$ that are ``bottom nodes'', and each such node $w$ is assumed to store all the corresponding runs $R$ such that $w$ is the bottommost node in $\Triangle(R)$.

For a general $\beta<|S|-1$, we use iterative range minimum queries to generate runs $R$ such that $T[i \dd i+|S|) \in \Triangle(R)$ in the order of non-decreasing periods.

We proceed to a formal implementation of this intuition.

\paragraph{Data structure.}
The data structure contains the suffix tree $\ST(T)$ and the tree of the suffix links $\ST'(T)$, both stored over the same set of nodes.
If the size of the alphabet of $T$ is constant, explicit children of an explicit node are stored in a list; otherwise, they are stored using perfect hashing~\cite{DBLP:journals/jacm/FredmanKS84} by first characters of edges leading to them.
Each node $v$ of $\ST(T)$ stores the smallest period of a run $R$ such that $v$ matches a fragment contained in the set $\Lower(R)$:
\[\MinLower(v)=\min(\{\infty\} \cup \{\per(R)\,:\,\#_v(\Lower(R))>0,\,R \in \R(T)\}).\]
For a run $R$ in $T$, we denote string label of the bottommost node in $\Triangle(R)$ as $\Bottom(R)$:
\[\Bottom(a,b,p)=T[a \dd b-p].\]
Each node $v$ of $\ST'(T)$ stores the set $\Bottoms(v)$ of all runs $R$ in $T$ such that $\bar{v}=\Bottom(R)$ sorted by $\per(R)$ as well as the minimal such period, denoted as $\MinBottom(v)$. That is, if
\[\{R \in \R(T)\,:\,\bar{v}=\Bottom(R)\}\ =\ \{R_1,\ldots,R_t\}\]
and $\per(R_1) \le \per(R_2) \le \dots \le \per(R_t)$, then $\Bottoms(v)=(R_1,\ldots,R_t)$ and $\MinBottom(v)$ is equal to $\per(R_1)$ if $t \ge 1$ and $\infty$ otherwise.

In each of the trees $\ST(T)$ and $\ST'(T)$, we number the explicit nodes with consecutive positive integers according to pre-order. Let $N(v)$, $N'(v)$ be the numbers assigned to $v$ in the respective trees. Each explicit node $v$ of $\ST(T)$ stores the intervals $I(v)$, $I'(v)$ of numbers of all explicit nodes in its subtree in $\ST(T)$ and in $\ST'(T)$, respectively. We construct arrays containing values $\MinLower(v)$ and $\MinBottom(v)$ according to the pre-order numberings:
\[
    \ML[i] = \MinLower(v)\text{ where }N(v)=i,\quad
    \MB[i] = \MinBottom(v)\text{ where }N'(v)=i
\]
and construct data structures for range minimum queries (RMQ)~\cite{DBLP:conf/latin/BenderF00} on these two arrays.

\paragraph{Data structure size.}
The size of the data structure is $\Oh(n)$ as the trees $\ST(T),\ST'(T)$ (cf.\ \cref{thm:ST}), the data structure for perfect hashing~\cite{DBLP:journals/jacm/FredmanKS84}, set of runs (\cref{thm:runs}) and the RMQ data structure~\cite{DBLP:conf/latin/BenderF00} take $\Oh(n)$ space.

\paragraph{Queries.}
Given a pattern $S$ and threshold $\beta$, we find the locus of $S$ in $\ST(T)$. 
Let $u$ be the nearest explicit descendant of the locus.
We use repetitive RMQs on array $\ML$ to find all explicit nodes $v$ in the subtree of $u$ such that $\MinLower(v) \le \beta$.
Precisely, we first find the minimum of $\ML[a \dd b]$ where $I(u)=[a \dd b]$.
If the minimum is greater than $\beta$, no nodes are reported.
Otherwise, we report a node $v$ for which $N(v)=k$ and $\ML[k]=\min\{\ML[a],\ldots,\ML[b]\}$ and recursively ask RMQs in $\ML[a \dd k-1]$ and in $\ML[k+1 \dd b]$.

For each node $v$ in the subtree of $u$ with $\MinLower[v] \le \beta$, we use repetitive RMQs on array $\MB$ to find all explicit nodes $w$ in the subtree of $v$ in $\ST'(T)$ such that $\MinBottom(w) \le \beta$.
Finally, for each such node $w$, we list all runs $R \in \Bottoms(w)$ such that $\per(R) \le \beta$.
Each reported run $R=(a,b,p)$ produces an overlapping consecutive occurrence $(i,j) \in \OOcc(S)$ such that $i=a+|\bar{w}|-|\bar{v}|$ and $j=i+p$.

\paragraph{Correctness of queries.} First we show that if $(i,j) \in \OOcc(S)$ and $j-i \le \beta$, then the algorithm reports $(i,j)$. By \cref{lem:triangle}, there exists a run $R=(a,b,p) \in \R(T)$ with period $p=j-i$ such that $T[i \dd i+m) \in \Triangle(R)$. Hence, $T[i \dd b-p] \in \Lower(R)$. Let $u$ and $v$ be the loci of $T[i \dd i+m)$ and $T[i \dd b-p]$, respectively, in $\ST(T)$. Then $v$ is an explicit descendant of $u$ (cf.\ \cref{{obs:lower_upper}}) and $\MinLower(v) \le p \le \beta$. Hence, $v$ will be reported in one of the RMQs in the first phase of the query algorithm. Let $w$ be the locus of $\Bottom(R)=T[a \dd b-p]$ in $\ST(T)$. Then $w$ is an explicit descendant of $v$ in $\ST'(T)$ and $\MinBottom(w) \le p \le \beta$. Hence, $w$ will be reported in the second phase of the query algorithm. Consequently, the run $R \in \Bottoms(w)$ will be reported, as $\per(R)=p \le \beta$. Finally, we have $a+|\bar{w}|-|\bar{v}|=a+(b-p-a+1)-(b-p-i+1)=i$ and $i+p=j$, so $(i,j)$ will indeed be reported in the query algorithm.

Now we show that if the algorithm reports $(i,j)$, then $(i,j) \in \OOcc(S)$ and $j-i \le \beta$. The second condition is obvious from the query algorithm, so let us focus on the first condition. Let us trace back how $(i,j)$ was obtained: let $u$ be the locus of $S$ in $\ST(T)$, $v$ be an explicit descendant of $u$ in $\ST(T)$ such that $\MinLower(v)\le \beta$, $w$ be an explicit descendant of $v$ in $\ST'(T)$ such that $\MinBottom(w) \le \beta$, and $R=(a,b,p) \in \Bottoms(w)$ be a run such that $i=a+|\bar{w}|-|\bar{v}|$ and $j=i+p$. There exist strings $U,V$ such that $SU$ is the string label of $v$ and $VSU$ is the string label of $w$. Further, $|\bar{w}|-|\bar{v}|=|V|$. Then $VSU=\Bottom(R)=T[a \dd b-p]$, so $S=T[a+|V| \dd a+|VS|)\equiv T[i \dd i+m)$. We have $T[i \dd i+m) \in \Triangle(R)$ as $i \ge a$, $i+m-1 \le b-p$, and $m>\beta \ge p$, so by \cref{lem:triangle}, $(i,j) \in \OOcc(S)$.

Finally, we show that each pair $(i,j) \in \OOcc(S)$ such that $j-i \le \beta$ is reported only once. Assume to the contrary that $(i,j)$ is reported twice, for runs $R_1$ and $R_2$. We have $\per(R_1)=\per(R_2)=j-i$ and each of the runs $R_1$, $R_2$ extends the fragment $T[i \dd i+m)$, where $m>\per(R_1)=\per(R_2)$. Hence, we must have $R_1=R_2=R=(a,b,p)$. This means that $(i,j)$ is reported for two different explicit descendants $v_1$, $v_2$ of $\locus(S)$, such that $\#_{v_1}(R), \#_{v_2}(R)>0$. For each of them, the bottom node $w$ is the same, as it is uniquely determined by $\Bottom(R)$. We have $i=a+|\bar{w}|-|\overline{v_1}|=a+|\bar{w}|-|\overline{v_2}|$, so $|\overline{v_1}|=|\overline{v_2}|$. This implies that $v_1=v_2$, as all fragments in the set $\Lower(R)$ clearly have different lengths; a contradiction.

\paragraph{Query complexity.}
Finding the locus of $S$ takes $\Oh(m)$ worst-case time using lists if the alphabet is constant and perfect hashing~\cite{DBLP:journals/jacm/FredmanKS84} otherwise. Each RMQ takes $\Oh(1)$ time~\cite{DBLP:conf/latin/BenderF00} and the number of queries asked is $\Oh(\output)$. The query time is $\Oh(m+\output)$ as desired.

\paragraph{Construction algorithm and complexity.} The suffix tree of $T$ can be constructed in $\Oh(n)$ time by \cref{thm:ST}. If the alphabet of $T$ is superconstant, the data structure for perfect hashing~\cite{DBLP:journals/jacm/FredmanKS84} requires $\Oh(n)$ expected time construction. All runs in $T$ can be generated in $\Oh(n)$ time by \cref{thm:runs}. To construct the sets $\Bottoms(v)$, we sort all runs by non-decreasing periods using bucket sort, find the loci of fragments $\Bottom(R)$ using weighted ancestor queries of \cref{thm:WAQ}, and insert the runs to the respective sets $\Bottoms(v)$ in order. This also allows to compute $\MinBottom(v)$ for each node $v$ of $\ST'(T)$. Now the pre-order traversal of $\ST(T)$ and $\ST'(T)$ as well as computation of arrays $\MB$ and $\ML$ takes $\Oh(n)$ time provided that values $\MinLower(v)$ for explicit nodes $v$ can be computed efficiently. Preprocessing for RMQ takes $\Oh(n)$ time~\cite{DBLP:conf/latin/BenderF00}.

Finally, to compute $\MinLower(v)$, we ask an RMQ on $\MB[a \dd b]$ where $I'(v)=[a \dd b]$. Let $p$ be the result of this query. If $|\bar{v}|>p$, we set $\MinLower(v)=p$. Otherwise, $\#_v(\Lower(R))=0$ for all runs $R$ in $T$, and therefore $\MinLower(v)=\infty$.

\medskip
We obtain the following result that was already stated in \cref{sec:intro}.
\thmCons*

\section{Conclusions}\label{sec:concl}
We have designed the first linear-time algorithm computing the Cover Suffix Tree. We have shown several applications of this result, some of which follow directly from previous work. Experimental comparison of our algorithms for computing the Cover Suffix Tree and the set of seeds in a string against implementations of existing methods from~\cite{DBLP:journals/tcs/CzajkaR21} is left as future work.

It remains an open problem if our approach can help to improve upon the $\Oh(n \log n)$-time algorithm of Brodal et al.~\cite{DBLP:conf/icalp/BrodalLOP02} for constructing MAST.

\subsection*{Acknowledgements}
The author thanks the anonymous reviewers for reading the manuscript carefully and providing several useful suggestions.

\bibliographystyle{plainurl}
\bibliography{references2}

\begin{thebibliography}{10}

\bibitem{DBLP:journals/talg/AmirLLS07}
Amihood Amir, Gad~M. Landau, Moshe Lewenstein, and Dina Sokol.
\newblock Dynamic text and static pattern matching.
\newblock {\em {ACM} Transactions on Algorithms}, 3(2):19, 2007.
\newblock URL: \url{https://doi.org/10.1145/1240233.1240242}, \href
  {http://dx.doi.org/10.1145/1240233.1240242}
  {\path{doi:10.1145/1240233.1240242}}.

\bibitem{DBLP:journals/tcs/ApostolicoE93}
Alberto Apostolico and Andrzej Ehrenfeucht.
\newblock Efficient detection of quasiperiodicities in strings.
\newblock {\em Theoretical Computer Science}, 119(2):247--265, 1993.
\newblock URL: \url{https://doi.org/10.1016/0304-3975(93)90159-Q}, \href
  {http://dx.doi.org/10.1016/0304-3975(93)90159-Q}
  {\path{doi:10.1016/0304-3975(93)90159-Q}}.

\bibitem{DBLP:journals/ipl/ApostolicoFI91}
Alberto Apostolico, Martin Farach, and Costas~S. Iliopoulos.
\newblock Optimal superprimitivity testing for strings.
\newblock {\em Information Processing Letters}, 39(1):17--20, 1991.
\newblock URL: \url{https://doi.org/10.1016/0020-0190(91)90056-N}, \href
  {http://dx.doi.org/10.1016/0020-0190(91)90056-N}
  {\path{doi:10.1016/0020-0190(91)90056-N}}.

\bibitem{DBLP:journals/algorithmica/ApostolicoP96}
Alberto Apostolico and Franco~P. Preparata.
\newblock Data structures and algorithms for the string statistics problem.
\newblock {\em Algorithmica}, 15(5):481--494, 1996.
\newblock URL: \url{https://doi.org/10.1007/BF01955046}, \href
  {http://dx.doi.org/10.1007/BF01955046} {\path{doi:10.1007/BF01955046}}.

\bibitem{DBLP:journals/siamcomp/BannaiIINTT17}
Hideo Bannai, Tomohiro I, Shunsuke Inenaga, Yuto Nakashima, Masayuki Takeda,
  and Kazuya Tsuruta.
\newblock The "runs" theorem.
\newblock {\em {SIAM} Journal on Computing}, 46(5):1501--1514, 2017.
\newblock URL: \url{https://doi.org/10.1137/15M1011032}, \href
  {http://dx.doi.org/10.1137/15M1011032} {\path{doi:10.1137/15M1011032}}.

\bibitem{DBLP:conf/cpm/BannaiIK17}
Hideo Bannai, Shunsuke Inenaga, and Dominik K{\"{o}}ppl.
\newblock Computing all distinct squares in linear time for integer alphabets.
\newblock In Juha K{\"{a}}rkk{\"{a}}inen, Jakub Radoszewski, and Wojciech
  Rytter, editors, {\em 28th Annual Symposium on Combinatorial Pattern
  Matching, {CPM} 2017}, volume~78 of {\em LIPIcs}, pages 22:1--22:18. Schloss
  Dagstuhl - Leibniz-Zentrum f{\"{u}}r Informatik, 2017.
\newblock URL: \url{https://doi.org/10.4230/LIPIcs.CPM.2017.22}, \href
  {http://dx.doi.org/10.4230/LIPIcs.CPM.2017.22}
  {\path{doi:10.4230/LIPIcs.CPM.2017.22}}.

\bibitem{DBLP:conf/cpm/BelazzouguiKPR21}
Djamal Belazzougui, Dmitry Kosolobov, Simon~J. Puglisi, and Rajeev Raman.
\newblock Weighted ancestors in suffix trees revisited.
\newblock In Paweł Gawrychowski and Tatiana Starikovskaya, editors, {\em 32nd
  Annual Symposium on Combinatorial Pattern Matching, {CPM} 2021}, volume 191
  of {\em LIPIcs}, pages 8:1--8:15. Schloss Dagstuhl - Leibniz-Zentrum
  f{\"{u}}r Informatik, 2021.
\newblock URL: \url{https://doi.org/10.4230/LIPIcs.CPM.2021.8}, \href
  {http://dx.doi.org/10.4230/LIPIcs.CPM.2021.8}
  {\path{doi:10.4230/LIPIcs.CPM.2021.8}}.

\bibitem{DBLP:conf/latin/BenderF00}
Michael~A. Bender and Martin Farach{-}Colton.
\newblock The {LCA} problem revisited.
\newblock In Gaston~H. Gonnet, Daniel Panario, and Alfredo Viola, editors, {\em
  4th Latin American Symposium on Theoretical Informatics, {LATIN} 2000},
  volume 1776 of {\em Lecture Notes in Computer Science}, pages 88--94.
  Springer, 2000.
\newblock URL: \url{https://doi.org/10.1007/10719839\_9}, \href
  {http://dx.doi.org/10.1007/10719839\_9} {\path{doi:10.1007/10719839\_9}}.

\bibitem{brlek2022number}
Srečko Brlek and Shuo Li.
\newblock On the number of squares in a finite word, 2022.
\newblock \href {http://arxiv.org/abs/2204.10204} {\path{arXiv:2204.10204}}.

\bibitem{DBLP:conf/cwords/BrlekL23}
Srečko Brlek and Shuo Li.
\newblock On the number of distinct squares in finite sequences: Some old and
  new results.
\newblock In Anna~E. Frid and Robert Mercas, editors, {\em 14th International
  Conference on Combinatorics on Words, {WORDS} 2023}, volume 13899 of {\em
  Lecture Notes in Computer Science}, pages 35--44. Springer, 2023.
\newblock URL: \url{https://doi.org/10.1007/978-3-031-33180-0\_3}, \href
  {http://dx.doi.org/10.1007/978-3-031-33180-0\_3}
  {\path{doi:10.1007/978-3-031-33180-0\_3}}.

\bibitem{DBLP:conf/icalp/BrodalLOP02}
Gerth~Stolting Brodal, Rune~B. Lyngso, Anna Ostlin, and Christian N.~S.
  Pedersen.
\newblock Solving the string statistics problem in time {O}(n log n).
\newblock In Peter Widmayer, Francisco~Triguero Ruiz, Rafael~Morales Bueno,
  Matthew Hennessy, Stephan~J. Eidenbenz, and Ricardo Conejo, editors, {\em
  29th International Colloquium on Automata, Languages and Programming, {ICALP}
  2002}, volume 2380 of {\em Lecture Notes in Computer Science}, pages
  728--739. Springer, 2002.
\newblock URL: \url{https://doi.org/10.1007/3-540-45465-9\_62}, \href
  {http://dx.doi.org/10.1007/3-540-45465-9\_62}
  {\path{doi:10.1007/3-540-45465-9\_62}}.

\bibitem{DBLP:conf/cpm/BrodalP00}
Gerth~Stolting Brodal and Christian N.~S. Pedersen.
\newblock Finding maximal quasiperiodicities in strings.
\newblock In Raffaele Giancarlo and David Sankoff, editors, {\em 11th Annual
  Symposium on Combinatorial Pattern Matching, {CPM} 2000}, volume 1848 of {\em
  Lecture Notes in Computer Science}, pages 397--411. Springer, 2000.
\newblock URL: \url{https://doi.org/10.1007/3-540-45123-4\_33}, \href
  {http://dx.doi.org/10.1007/3-540-45123-4\_33}
  {\path{doi:10.1007/3-540-45123-4\_33}}.

\bibitem{DBLP:journals/jacm/BrownT79}
Mark~R. Brown and Robert~Endre Tarjan.
\newblock A fast merging algorithm.
\newblock {\em Journal of the {ACM}}, 26(2):211--226, 1979.
\newblock URL: \url{https://doi.org/10.1145/322123.322127}, \href
  {http://dx.doi.org/10.1145/322123.322127} {\path{doi:10.1145/322123.322127}}.

\bibitem{DBLP:books/daglib/0020103}
Maxime Crochemore, Christophe Hancart, and Thierry Lecroq.
\newblock {\em Algorithms on strings}.
\newblock Cambridge University Press, 2007.

\bibitem{DBLP:journals/tcs/CrochemoreIKRRW14}
Maxime Crochemore, Costas~S. Iliopoulos, Marcin Kubica, Jakub Radoszewski,
  Wojciech Rytter, and Tomasz Waleń.
\newblock Extracting powers and periods in a word from its runs structure.
\newblock {\em Theoretical Computer Science}, 521:29--41, 2014.
\newblock URL: \url{https://doi.org/10.1016/j.tcs.2013.11.018}, \href
  {http://dx.doi.org/10.1016/j.tcs.2013.11.018}
  {\path{doi:10.1016/j.tcs.2013.11.018}}.

\bibitem{DBLP:journals/tcs/CzajkaR21}
Patryk Czajka and Jakub Radoszewski.
\newblock Experimental evaluation of algorithms for computing quasiperiods.
\newblock {\em Theoretical Computer Science}, 854:17--29, 2021.
\newblock URL: \url{https://doi.org/10.1016/j.tcs.2020.11.033}, \href
  {http://dx.doi.org/10.1016/j.tcs.2020.11.033}
  {\path{doi:10.1016/j.tcs.2020.11.033}}.

\bibitem{DBLP:conf/icalp/Ellert021}
Jonas Ellert and Johannes Fischer.
\newblock Linear time runs over general ordered alphabets.
\newblock In Nikhil Bansal, Emanuela Merelli, and James Worrell, editors, {\em
  48th International Colloquium on Automata, Languages, and Programming,
  {ICALP} 2021}, volume 198 of {\em LIPIcs}, pages 63:1--63:16. Schloss
  Dagstuhl - Leibniz-Zentrum f{\"{u}}r Informatik, 2021.
\newblock URL: \url{https://doi.org/10.4230/LIPIcs.ICALP.2021.63}, \href
  {http://dx.doi.org/10.4230/LIPIcs.ICALP.2021.63}
  {\path{doi:10.4230/LIPIcs.ICALP.2021.63}}.

\bibitem{DBLP:conf/focs/Farach97}
Martin Farach.
\newblock Optimal suffix tree construction with large alphabets.
\newblock In {\em 38th Annual Symposium on Foundations of Computer Science,
  {FOCS} 1997}, pages 137--143. {IEEE} Computer Society, 1997.
\newblock URL: \url{https://doi.org/10.1109/SFCS.1997.646102}, \href
  {http://dx.doi.org/10.1109/SFCS.1997.646102}
  {\path{doi:10.1109/SFCS.1997.646102}}.

\bibitem{fine1965uniqueness}
Nathan~J. Fine and Herbert~S. Wilf.
\newblock Uniqueness theorems for periodic functions.
\newblock {\em Proceedings of the American Mathematical Society},
  16(1):109--114, 1965.
\newblock \href {http://dx.doi.org/10.2307/2034009}
  {\path{doi:10.2307/2034009}}.

\bibitem{DBLP:journals/jct/FraenkelS98}
Aviezri~S. Fraenkel and Jamie Simpson.
\newblock How many squares can a string contain?
\newblock {\em Journal of Combinatorial Theory, Series {A}}, 82(1):112--120,
  1998.
\newblock URL: \url{https://doi.org/10.1006/jcta.1997.2843}, \href
  {http://dx.doi.org/10.1006/jcta.1997.2843}
  {\path{doi:10.1006/jcta.1997.2843}}.

\bibitem{DBLP:journals/jacm/FredmanKS84}
Michael~L. Fredman, J{\'{a}}nos Koml{\'{o}}s, and Endre Szemer{\'{e}}di.
\newblock Storing a sparse table with {O}(1) worst case access time.
\newblock {\em Journal of the {ACM}}, 31(3):538--544, 1984.
\newblock URL: \url{https://doi.org/10.1145/828.1884}, \href
  {http://dx.doi.org/10.1145/828.1884} {\path{doi:10.1145/828.1884}}.

\bibitem{DBLP:conf/soda/GanardiG22}
Moses Ganardi and Paweł Gawrychowski.
\newblock Pattern matching on grammar-compressed strings in linear time.
\newblock In Joseph~(Seffi) Naor and Niv Buchbinder, editors, {\em Proceedings
  of the 2022 {ACM-SIAM} Symposium on Discrete Algorithms, {SODA} 2022}, pages
  2833--2846. {SIAM}, 2022.
\newblock URL: \url{https://doi.org/10.1137/1.9781611977073.110}, \href
  {http://dx.doi.org/10.1137/1.9781611977073.110}
  {\path{doi:10.1137/1.9781611977073.110}}.

\bibitem{DBLP:conf/esa/GawrychowskiLN14}
Paweł Gawrychowski, Moshe Lewenstein, and Patrick~K. Nicholson.
\newblock Weighted ancestors in suffix trees.
\newblock In Andreas~S. Schulz and Dorothea Wagner, editors, {\em 22th Annual
  European Symposium on Algorithms, Wrocław, Poland, {ESA} 2014}, volume 8737
  of {\em Lecture Notes in Computer Science}, pages 455--466. Springer, 2014.
\newblock URL: \url{https://doi.org/10.1007/978-3-662-44777-2\_38}, \href
  {http://dx.doi.org/10.1007/978-3-662-44777-2\_38}
  {\path{doi:10.1007/978-3-662-44777-2\_38}}.

\bibitem{DBLP:journals/ipl/Hershberger89}
John Hershberger.
\newblock Finding the upper envelope of n line segments in {O}(n log n) time.
\newblock {\em Information Processing Letters}, 33(4):169--174, 1989.
\newblock URL: \url{https://doi.org/10.1016/0020-0190(89)90136-1}, \href
  {http://dx.doi.org/10.1016/0020-0190(89)90136-1}
  {\path{doi:10.1016/0020-0190(89)90136-1}}.

\bibitem{DBLP:journals/algorithmica/IliopoulosMP96}
Costas~S. Iliopoulos, Dennis W.~G. Moore, and Kunsoo Park.
\newblock Covering a string.
\newblock {\em Algorithmica}, 16(3):288--297, 1996.
\newblock URL: \url{https://doi.org/10.1007/BF01955677}, \href
  {http://dx.doi.org/10.1007/BF01955677} {\path{doi:10.1007/BF01955677}}.

\bibitem{DBLP:journals/jacm/KarkkainenSB06}
Juha K{\"{a}}rkk{\"{a}}inen, Peter Sanders, and Stefan Burkhardt.
\newblock Linear work suffix array construction.
\newblock {\em Journal of the {ACM}}, 53(6):918--936, 2006.
\newblock URL: \url{https://doi.org/10.1145/1217856.1217858}, \href
  {http://dx.doi.org/10.1145/1217856.1217858}
  {\path{doi:10.1145/1217856.1217858}}.

\bibitem{DBLP:conf/soda/KociumakaKRRW12}
Tomasz Kociumaka, Marcin Kubica, Jakub Radoszewski, Wojciech Rytter, and Tomasz
  Waleń.
\newblock A linear time algorithm for seeds computation.
\newblock In Yuval Rabani, editor, {\em Proceedings of the Twenty-Third Annual
  {ACM-SIAM} Symposium on Discrete Algorithms, {SODA} 2012}, pages 1095--1112.
  {SIAM}, 2012.
\newblock URL: \url{https://doi.org/10.1137/1.9781611973099.86}, \href
  {http://dx.doi.org/10.1137/1.9781611973099.86}
  {\path{doi:10.1137/1.9781611973099.86}}.

\bibitem{DBLP:journals/talg/KociumakaKRRW20}
Tomasz Kociumaka, Marcin Kubica, Jakub Radoszewski, Wojciech Rytter, and Tomasz
  Waleń.
\newblock A linear-time algorithm for seeds computation.
\newblock {\em {ACM} Transations on Algorithms}, 16(2):27:1--27:23, 2020.
\newblock URL: \url{https://doi.org/10.1145/3386369}, \href
  {http://dx.doi.org/10.1145/3386369} {\path{doi:10.1145/3386369}}.

\bibitem{DBLP:journals/algorithmica/KociumakaPRRW15}
Tomasz Kociumaka, Solon~P. Pissis, Jakub Radoszewski, Wojciech Rytter, and
  Tomasz Waleń.
\newblock Fast algorithm for partial covers in words.
\newblock {\em Algorithmica}, 73(1):217--233, 2015.
\newblock URL: \url{https://doi.org/10.1007/s00453-014-9915-3}, \href
  {http://dx.doi.org/10.1007/s00453-014-9915-3}
  {\path{doi:10.1007/s00453-014-9915-3}}.

\bibitem{DBLP:journals/tcs/KociumakaPRRW18a}
Tomasz Kociumaka, Solon~P. Pissis, Jakub Radoszewski, Wojciech Rytter, and
  Tomasz Waleń.
\newblock Efficient algorithms for shortest partial seeds in words.
\newblock {\em Theoretical Computer Science}, 710:139--147, 2018.
\newblock URL: \url{https://doi.org/10.1016/j.tcs.2016.11.035}, \href
  {http://dx.doi.org/10.1016/j.tcs.2016.11.035}
  {\path{doi:10.1016/j.tcs.2016.11.035}}.

\bibitem{DBLP:conf/focs/KolpakovK99}
Roman~M. Kolpakov and Gregory Kucherov.
\newblock Finding maximal repetitions in a word in linear time.
\newblock In {\em 40th Annual Symposium on Foundations of Computer Science,
  {FOCS} 1999}, pages 596--604. {IEEE} Computer Society, 1999.
\newblock URL: \url{https://doi.org/10.1109/SFFCS.1999.814634}, \href
  {http://dx.doi.org/10.1109/SFFCS.1999.814634}
  {\path{doi:10.1109/SFFCS.1999.814634}}.

\bibitem{DBLP:journals/jacm/McCreight76}
Edward~M. McCreight.
\newblock A space-economical suffix tree construction algorithm.
\newblock {\em Journal of the {ACM}}, 23(2):262--272, 1976.
\newblock URL: \url{https://doi.org/10.1145/321941.321946}, \href
  {http://dx.doi.org/10.1145/321941.321946} {\path{doi:10.1145/321941.321946}}.

\bibitem{DBLP:conf/soda/MooreS94}
Dennis W.~G. Moore and William~F. Smyth.
\newblock Computing the covers of a string in linear time.
\newblock In Daniel~Dominic Sleator, editor, {\em Proceedings of the Fifth
  Annual {ACM-SIAM} Symposium on Discrete Algorithms. SODA 1994}, pages
  511--515. {ACM/SIAM}, 1994.
\newblock URL: \url{http://dl.acm.org/citation.cfm?id=314464.314636}.

\bibitem{DBLP:journals/ipl/MooreS95}
Dennis W.~G. Moore and William~F. Smyth.
\newblock A correction to "{A}n optimal algorithm to compute all the covers of
  a string".
\newblock {\em Information Processing Letters}, 54(2):101--103, 1995.
\newblock URL: \url{https://doi.org/10.1016/0020-0190(94)00235-Q}, \href
  {http://dx.doi.org/10.1016/0020-0190(94)00235-Q}
  {\path{doi:10.1016/0020-0190(94)00235-Q}}.

\bibitem{DBLP:journals/tcs/NavarroT16}
Gonzalo Navarro and Sharma~V. Thankachan.
\newblock Reporting consecutive substring occurrences under bounded gap
  constraints.
\newblock {\em Theoretical Computer Science}, 638:108--111, 2016.
\newblock URL: \url{https://doi.org/10.1016/j.tcs.2016.02.005}, \href
  {http://dx.doi.org/10.1016/j.tcs.2016.02.005}
  {\path{doi:10.1016/j.tcs.2016.02.005}}.

\end{thebibliography}

\end{document}